\documentclass[11pt, letterpaper]{article}

\usepackage{amsmath,amssymb,amsthm}
\usepackage[numbers]{natbib}
\usepackage{graphicx}
\usepackage{graphicx,colordvi,psfrag}
\usepackage{amsmath,amssymb,algpseudocode}
\usepackage{algorithm}
\usepackage{algorithmicx}
\usepackage{calc,pstricks, pgf, xcolor}
\usepackage{bbding}
\usepackage{bbm}
\usepackage{dsfont}
\usepackage{stfloats}

\usepackage{xparse}
\usepackage{enumerate}
\usepackage{tikz}
\usepackage{mathtools}
\usetikzlibrary{arrows.meta, positioning}

\usepackage{subcaption}
\usepackage{amssymb}

\usepackage[margin=1in]{geometry}
\usepackage[
            CJKbookmarks=true,
            bookmarksnumbered=true,
            bookmarksopen=true,
            colorlinks=true,
            citecolor=red,
            linkcolor=blue,
            anchorcolor=red,
            urlcolor=blue,
            ]{hyperref}

\hypersetup{
  colorlinks   = true, 
  urlcolor     = blue, 
  linkcolor    = blue, 
  citecolor   = red 
}

\newtheorem{theorem}{Theorem}
\newtheorem{remark}{Remark}[section]
\newtheorem{lemma}{Lemma}[section]
\newtheorem{claim}{Claim}[section]
\newtheorem{corollary}{Corollary}[section]

\newtheorem{proposition}{Proposition}[section]
\newtheorem{fact}{Fact}[section]
\theoremstyle{definition}

\newtheorem{condition}{Condition}

\newtheorem{innercustomlem}{Lemma}
\newenvironment{mylem}[1]
  {\renewcommand\theinnercustomlem{#1}\innercustomlem\itshape}
  {\endinnercustomlem}

\newtheorem{innercustomcond}{Condition}
\newenvironment{mycond}[1]
  {\renewcommand\theinnercustomcond{#1}\innercustomcond\itshape}
  {\endinnercustomcond}

  \newtheorem{innercustomthm}{Theorem}
\newenvironment{mythm}[1]
  {\renewcommand\theinnercustomthm{#1}\innercustomthm\itshape}
  {\endinnercustomthm}

\newcommand{\E}{\mathbb{E}}
\newcommand{\R}{\mathbb{R}}
\newcommand{\Ind}{\mathds{1}}
\newcommand{\I}{\mathbbm{1}}

\newcommand{\ZZ}{\mathbb{Z}}
\newcommand{\RR}{\mathbb{R}}

\newcommand{\PP}{\mathbb{P}}

\newcommand{\Var}{\mathrm{Var}}

\newcommand{\diag}{\mathop{\mathrm{diag}}}

\newcommand{\spec}{\mathop{\mathrm{spec}}}

\newcommand{\trace}{\mathop{\mathrm{tr}}}
\def\tr{\trace}

\newcommand{\eps}{\varepsilon}

\def\mreals{\mathbb{R}}

\def\EE{\mathbb{E}}
\def\simiid{\stackrel{\mathrm{iid}}{\sim}}
\newcommand{\m}{\mathcal}

\newcommand{\suc}{\mathrm{success}}

\newcommand{\paren}[1]{\left( #1 \right)}
\newcommand{\sqb}[1]{\left[ #1 \right]}
\newcommand{\sset}[1]{\left\{ #1 \right\}}
\newcommand{\B}{\Big}
\renewcommand{\b}{\big}

\NewDocumentCommand{\DIP}{e{^_}}{D^{\mathrm{IP}\IfValueT{#1}{,#1}}_{\IfValueT{#2}{#2}}}

\NewDocumentCommand{\DMM}{e{^_}}{D^{\mathrm{MM}\IfValueT{#1}{,#1}}_{\IfValueT{#2}{#2}}}

\def\calN{\mathcal{N}}
\def\N{\calN}

\newcommand{\Ut}{\widetilde{U}}
\newcommand{\Sigmat}{\widetilde{\Sigma}}
\newcommand{\Wt}{\widetilde{W}}

\newcommand{\ct}{\tilde{c}}
\newcommand{\taut}{\widetilde{\tau}}
\newcommand{\xt}{\widetilde{x}}
\newcommand{\Lambdat}{\widetilde{\Lambda}}
\newcommand{\lambdat}{\widetilde{\lambda}}
\newcommand{\mut}{\widetilde{\mu}}
\newcommand{\ttilde}{\widetilde{t}}
\newcommand{\Ttilde}{\widetilde{T}}

\newcommand{\Wh}{\widehat{W}}
\newcommand{\Yh}{\widehat{Y}}

\newcommand{\sh}{\widehat{s}}
\newcommand{\dhat}{\widehat{d}}
\newcommand{\Th}{\widehat{T}}
\newcommand{\Ah}{\widehat{A}}
\newcommand{\Bh}{\widehat{B}}

\newcommand{\one}{\mathrm{I}}
\newcommand{\two}{\mathrm{II}}

\newcommand{\CCC}{\mathbf{C}}

\newcommand{\EVD}{\mathrm{EVD}}
\newcommand{\poly}{\mathrm{poly}}

\newcommand{\Rstar}{R^{\star}}
\newcommand{\Dstar}{D^{\star}}
\newcommand{\DDrc}{\mathbf{D}_{\mathrm{rc}}}
\newcommand{\DDwf}{\mathbf{D}_{\mathrm{wf}}}
\newcommand{\RRrc}{\mathbf{R}_{\mathrm{rc}}}
\newcommand{\RRwf}{\mathbf{R}_{\mathrm{wf}}}

\newcommand{\twf}{t_{\mathrm{wf}}}
\newcommand{\Trc}{T_{\mathrm{rc}}}
\newcommand{\Rnat}{R_{\mathrm{nat}}}
\newcommand{\pstar}{p^{\star}}
\newcommand{\ttildewf}{\widetilde{t}_{\mathrm{wf}}}
\newcommand{\Ttilderc}{\widetilde{T}_{\mathrm{rc}}}

\newcommand{\calA}{\mathcal{A}}

\newcommand{\Drc}{D_{\mathrm{rc}}}
\newcommand{\Rrc}{R_{\mathrm{rc}}}
\newcommand{\Dwf}{D_{\mathrm{wf}}}
\newcommand{\Rwf}{R_{\mathrm{wf}}}

\renewcommand{\m}{\mathcal}

\DeclareMathOperator*{\argmin}{\arg\!\min}

\title{Price of metric universality in vector quantization is at most 0.11 bit}
\author{Alina Harbuzova\thanks{Massachusetts Institute of Technology. Supported by MathWorks Fellowship and Siebel Scholarship} \and Or Ordentlich\thanks{Hebrew University of Jerusalem. Supported by the
Israel Science Foundation (ISF), grant No. 2878/25} \and Yury Polyanskiy\thanks{Massachusetts Institute of Technology. Supported by NSF Grant No. 2112665 via subaward KR 704702 from the University of California, San Diego}}
\date{} 

\newenvironment{keywords}{\par\noindent\textbf{Keywords: }\ignorespaces}{\par}

\begin{document}
\pagenumbering{gobble}
\maketitle

\begin{abstract}%
Fast computation of a matrix product $W^\top X$ is a workhorse of modern LLMs. To make their
deployment more efficient, a popular approach is that of using a low-precision approximation $\widehat
W$ in place of true $W$ (``weight-only quantization''). Information theory demonstrates that an optimal algorithm  for reducing precision
of $W$ depends on the (second order) statistics of $X$ and requires a careful alignment of vector
quantization codebook with PCA directions of $X$ (a process known as ``waterfilling
allocation''). Dependence of the codebook on statistics of $X$, however, is highly impractical.
This paper proves that there exist a universal codebook that is
simultaneously near-optimal for all possible statistics of $X$, in the sense of being at least as
good as an $X$-adapted waterfilling codebook with rate reduced by 0.11 bit per dimension in the case when $W$ is Gaussian. Such
universal codebook would be an ideal candidate for the low-precision storage
format, a topic of active modern research, but alas the existence proof is non-constructive.

Equivalently, our result
shows existence of a net in $\mathbb{R}^n$ that is a nearly-optimal covering of a sphere
simultaneously with respect to all Hilbert norms. 
\end{abstract}




\newpage

\setcounter{tocdepth}{2}
\hypersetup{linktoc=all}
\tableofcontents

\newpage
\pagenumbering{arabic}
\setcounter{page}{1}
\section{Introduction}

The most basic element of all modern AI is a neural unit: given a (dynamically changing)
activation vector $X\in\mreals^n$ the unit needs to compute the output
$$ Y  = W^\top X\,,$$
where $W\in \mreals^n$ is a (static) weight. The problem, rapidly becoming central for economic deployment and continued
evolution of large language models (LLMs), is to reduce storage/communication requirement by saving $W$ in
``low-precision''. (The symmetric question of also converting $X$ to low-precision is outside of
scope of this work, though see \cite{ordentlich2025optimal} for some recent theoretical analysis.)

Early deep learning models saved $W$ in full-precision (known as FP32, and corresponding to $R=32$ bit
/ coordinate), but soon moved to half-precision (FP16/BF16, or $R=16$ bit). In the domain of LLMs,
pioneering work \cite{dettmers2022gpt3}  showed that very little degradation is
introduced if $W$ is approximated by rescaling it to appropriate range  and then rounding each coordinate of normalized $W$ to
nearest integer in $\{-128,-127,\ldots,127\}$, a so called \textit{INT8} quantization. Subsequently, more sophisticated low-precision storage formats were introduced,
with currently the most popular being NVFP4 (rate $R=4.5$ bit) and MXFP4 (rate $R=4.25$ bit),
see~\cite{nvidia2025nvfp4,ocp2023mx}.

In this paper we are focusing on a fundamental question: what is the best way of reducing
precision of $W$? That is, how to replace $W$ by a version $\Wh$ that incurs minimal
degradation of performance, i.e. $\Wh^\top X \approx W^\top X$, while admitting a short
bit-length description. A natural way to do that, known as \textit{vector quantization}, would be
to pre-define a \textit{codebook} $\CCC\subset \mreals^n$ of size $|\CCC|=2^{nR}$. Clearly, any element of $\CCC$ can be described
by $nR$ bits, hence achieving rate $R$ of bits / coordinate. Given $\CCC$ we approximate $W$ as
$$ \Wh = \argmin_{c\in\CCC} d(W, c)$$
for some distance metric $d(\cdot,\cdot)$. When $d$ is a standard Euclidean metric,
then the problem reduces to a classical vector quantization problem in $\mreals^n$, with many classical
solutions including lattices and trellis-coded constructions \cite{gersho2012vector}. However, as was brilliantly shown by \cite{frantar2023gptq} large savings can be made if metric $d(\cdot,\cdot)$ is chosen with the knowledge
of statistics of $X$ in mind. 

Indeed, if $X$ is modeled as random with second-order statistics $\Sigma_X = \EE[X X^\top]$ then 
$$ \EE_X\sqb{(Y-\Yh)^2} = \EE_X\sqb{(W^\top X-\Wh^\top X)^2} = (W-\Wh)^\top \Sigma_X (W-\Wh)\,.$$
Thus, we see that a natural choice of metric (given knowledge of $\Sigma_X$) is 
\begin{equation}\label{eq:dist_intro}
d_{\Sigma_X}(W, \Wh) = \EE_X\sqb{(W^\top X-\Wh^\top X)^2} = (W- \Wh)^\top \Sigma_X (W - \Wh)\,.\tag{$d_{\Sigma_X}$}
\end{equation}
The easiest way to demonstrate how adaptation to $\Sigma_X$ can significantly improve rate-distortion
tradeoff is to consider a rank-1 case, i.e. when $X$ is always collinear with a fixed vector
$v\in\mreals^n$. In this case, a clever choice of the codebook $\CCC$ is $\{0, \pm \epsilon
v, \pm 2 \epsilon v, \ldots\}$, i.e. very fine quantization along a single direction $v$ in
$\mreals^n$. Indeed, by not needing to spread the points of $\CCC$ among all $n$
dimensions, one can get exponential improvement in quality of approximation of $W^\top X$, since
only the scalar value $W^\top v$ affects the result. Since activations in LLMs are notoriously
low-rank, this adaptation of $\CCC$ to directions of principal variation (PCA) of $X$
understandably improves performance.

Herein, however, lies the main problem that we are trying to address: while adapting $\CCC$ to
the statistics of the input $X$ is desirable, it may not be generally possible due to restrictions
of hardware. Indeed, the mapping from actual bits (loaded from memory) to elements of $\CCC$ needs
to be fixed at hardware design stage and cannot depend on statistics of $X$ (in particular,
because the same hardware is used for implementing different neurons, facing different types of
$X$). Below we call this requirement, alternatively, as \textit{universal codebook} or 
\textit{a $\Sigma_X$-oblivious decoder}, to reflect the fact that $\CCC$ has to be universal
across all possible choices of $\Sigma_X$.

To continue with more quantitative investigation, let us make a modeling assumption (well
justified by empirical statistics of LLM weight matrices) that $W \sim \mathcal{N}(0,I_n)$. In
this case a given codebook $\CCC$ under $\Sigma_X$ statistics attains \textit{distortion}:
    \begin{equation}\label{eq:dist_ccc}
    	D(\CCC, \Sigma_X) := \frac 1 n \E_W \sqb{\min_{c\in\CCC} d_{\Sigma_X}(W, c)}\,. 
\end{equation}    
Classical information-theoretic field, known as rate-distortion theory, establishes that for a
codebook $\CCC$ to achieve distortion $D(\CCC,\Sigma_X)\le D$ one must have
$$ \log |\CCC| \ge n \RRwf(\Sigma_X,D)\,,$$
where $\RRwf$ is given by a so-called \textit{waterfilling formula}, see
Prop.~\ref{prop:oracle_wf}.

The main question of this work: \textit{How much does the requirement of universality cost in terms of
performance?} For example, in the rank-1 case above the waterfilling codebook would allocate its
elements along a single direction $v$. This codebook, however, would be grossly suboptimal for
another rank-1 $\Sigma_X$ which has its PCA direction orthogonal to $v$. Somewhat surprisingly,
thus, we show that nevertheless the answer is \textit{not much}.  
The main result of this work is demonstration of existence of a universal $\CCC$, which is
simultaneously near optimal for all possible $\Sigma_X$. Informally, we can state our main result
as follows.

\begin{theorem}[Informal: Universality costs $\le 0.11$ Bits]\label{thm:informal}
    Let us assume that $W\sim \mathcal{N}(0,I_n)$ and let $\RRwf(\Sigma_X, D)$ denote the
    information-theoretic (waterfilling) lower bound on rate needed to achieve distortion at most
    $D$ in the \emph{oracle} setting, where codebook is optimized for a fixed $\Sigma_X$. There
    exists a universal codebook $\CCC$ with $2^{nR}$ points such that its distortion
    \emph{simultaneously} for all $\Sigma_X \in \mathbb{S}_+^n$ satisfies:
    $$ R \le \RRwf (\Sigma_X, D(\CCC, \Sigma_X)) + 0.11 \, \mathrm{bit}\,. $$
\end{theorem}

The implication for hardware design is clear: it is possible to create a universal low-precision
storage format (for $W$) that is optimal (up to rate gap of at most $0.11$ bit) simultaneously for all kinds of distributions of
statistics of the other factor ($X$) in the inner-product. Our result can also be interpreted as a
statement about metric entropy: \textit{There exists a universal
net on a unit sphere, which covers unit sphere near-optimally simultaneously for all possible
Hilbert norms on $\mreals^n$.}

\paragraph{Paper organization.} The following Section~\ref{sec:main_results} formalizes weight-only quantization for inner products under the distortion~\ref{eq:dist_intro}. We then introduce the oracle benchmark given by the Gaussian rate--distortion tradeoff under weighted MSE, attained by the waterfilling solution in the setting where both encoder and decoder know the second-order statistics $\Sigma_X$ (Prop.~\ref{prop:oracle_wf}). Our main results are stated in Theorems~\ref{thm:union} and \ref{thm:worst_case}. Theorem~\ref{thm:union} shows the existence of a universal codebook achieving the explicit rate-distortion tradeoff~\eqref{eq:rd_rc_G_intro} over all $\Sigma_X$. Theorem~\ref{thm:worst_case} upper bounds the worst-case rate overhead incurred by this universal decoder relative to the oracle waterfilling benchmark. 

Section~\ref{sec:ProofSketch} provides a proof sketch of Theorem~\ref{thm:union} and the main geometric ideas: Section~\ref{subsec:ideas_codebook} describes the universal codebook construction and intuition, and Section~\ref{subsec:ideas_proof_sketch} outlines the random-coding analysis leading to~\eqref{eq:rd_rc_G_intro}.

Complete proofs are given in the Appendix. While the results in Section~\ref{sec:main_results} are presented for $W\sim \N(0, I_n)$, we first establish a general bound for a fixed (non-random) $W$ in Section~\ref{subsec:rc_general_W}. This result is then specialized to $W\sim \N(0,I_n)$ using concentration and covering argument in Section~\ref{sec:rc_gaussian_W}, completing the proof of Theorem~\ref{thm:union}. Theorem~\ref{thm:worst_case} is proved in Section~\ref{sec:worst_case}: we derive explicit expressions for the rate-gap and prove that the maximum gap occurs at spectra with at most $2$ distinct eigenvalues and vanishing distortions. 

\section{Main results and discussion}\label{sec:main_results}

Consider an arbitrary $\Sigma_X \succeq 0$ and define a (square of) Hilbert metric with respect to
$\Sigma_X$ as in~\eqref{eq:dist_intro}. 
We consider the problem of obtaining a low-precision (at rate $R$ bits per coordinate)
representation $\hat W$ of a random vector $W \in \mathbb{R}^{n}$ with the goal of minimizing
$d_{\Sigma_X}(\hat W,W)$. 
We focus presentation of results on the the standard setting in which  $W$ is an isotropic Gaussian vector:
\begin{equation*}\label{eq:Wmodel_intro}
    W \sim \N(0, I_n)\,,
\end{equation*}
though the key technical results hold for general $W$ (Section~\ref{subsec:rc_general_W}).
An $(n,R)$ quantization scheme consists of an encoder $f:\R^{n}\to \sqb{2^{nR}}$ and decoder
$g:\sqb{2^{nR}} \to \R^{n}$ and we set $\Wh = g(f(W))$. The image of $g$ is called  the
codebook $\CCC:=\mathrm{im}\, g$. The optimal encoder consists of finding a nearest to $W$ element of $\CCC$, and
hence we can equivalently think of a quantization scheme as completely defined by $\CCC$. The
distortion of $\CCC$ for a given $\Sigma_X$ is denoted $D(\CCC,\Sigma_X)$,
cf.~\eqref{eq:dist_ccc}.

Let us start with a simple case of $\Sigma_X$ fixed (and hence known to both encoder
$f_{\Sigma_X}$ and decoder $g_{\Sigma_X}$). In this case, the optimal tradeoff between the
distortion $D$ and rate $R$ is given by waterfilling, which we review.

Let $\Sigma_X = U\Lambda U^\top$ be the eigendecomposition with $\Lambda =
\diag(\lambda_1,\dots,\lambda_n)$. Consider a parametric curve
\begin{equation}\label{eq:wf_cruve}
\Dwf(\Sigma_X, t) = \frac{1}{n}\sum_{i=1}^n \min\sset{\lambda_i,t}\quad \Rwf(\Sigma_X, t) = \frac 1 {2n} \sum_{i=1}^n \max\sset{0, \log (\lambda_i/t)}\,,\tag{$\mathrm{WF}$}
\end{equation}
where $t$ is the parameter (waterfilling level). The above implicitly defines waterfilling distortion as a function of rate:
\begin{equation}\label{eq:wf_dist}
    \DDwf(R, \Sigma_X) \triangleq \Dwf(\Sigma_X, \twf(R))\,, \quad \text{where } \Rwf(\Sigma_X, \twf(R)) = R\,.\tag{$\DDwf$}
\end{equation}
It turns out that this function indeed determines the fundamental limits in the case of
oracle-knowledge of $\Sigma_X$. More exactly, we have the following.

\begin{proposition}[Waterfilling]\label{prop:oracle_wf}
Let $\Sigma_X \in \mathbb{S}_+^n$ . For any
compression scheme of rate $R$ we have
\begin{equation}\label{eq:lb_wf}
 \E\sqb{d_{\Sigma_X}(W,g(f(W)))} \ge n \DDwf(R,\Sigma_X)\,.
 \end{equation}
Conversely, for any $B>0$ there exist a $c=c(B)>0$ such that for any $\Sigma_X$ there exist $f$
and $g$ (both depending on $\Sigma_X$) such that 
\begin{equation}\label{eq:ub_wf}
	\E\sqb{d_{\Sigma_X}(W,g(f(W)))} \le n \DDwf\left(R-c\sqrt{\frac{\log n} n},\Sigma_X\right) + c n^{-B} \tr
\Sigma_X\,.
\end{equation}
\end{proposition}

\begin{proof}
We give the proof of the lower bound~\eqref{eq:lb_wf} below; since the upper bound~\eqref{eq:ub_wf} is not relevant for the
rest of this paper, we defer the brief of sketch of the proof to Appendix~\ref{sec:ub_wf}. Though
we do emphasize that the upper bound does not follow from classical  theory, which concerns with
separable (additive over coordinates) distortion measures, and we have to invoke more modern
single-shot bounds, cf.~\cite[Chapter 25]{PWbook24}.

Let $\Sigma_X = U\Lambda U^\top$ with $\Lambda = \diag(\lambda_1,\dots,\lambda_n)$, and define $W' \triangleq U^\top W$, $\Wh' \triangleq U^\top \Wh$. Then 
$$
d_{\Sigma_X}(W, \Wh) = (W - \Wh)^\top \Sigma_X (W- \Wh) = (W' - \Wh')^\top \Lambda (W' - \Wh') = \sum_{i=1}^n \lambda_i (W_i' - \Wh_i')^2\,.
$$
In the oracle setting (where both the encoder and decoder know $\Sigma_X$), we may equivalently compress $W'$, reconstruct $\Wh'$, and output $\Wh = U \Wh'$. After this coordinate change, the problem becomes that of a weighted mean squared error. Since $W'\sim
\mathcal{N}(0,I_n)$ has independent coordinates, a standard data-processing argument (see~\cite[Section 23.4 and Theorem 6.1]{PWbook24}) gives:
$$ nR \ge I(W';\Wh') \ge \sum_{i=1}^n I(W'_i; \Wh_i')\,,$$
Writing $D_i \triangleq \E \sqb{(W_i' - \Wh_i')^2}$, the smallest $I(W'_i; \hat W_i')={\frac 1 2} \log {\frac 1{D_i}}$ (given the value $D_i$) is then attained under Gaussian coupling, cf.~\cite[Section 26.1.2]{PWbook24}, which results in 
\begin{align*} R \ge {\frac 1 {2n}} \sum_{i=1}^n \log {\frac1 {D_i}} \qquad\text{and}\qquad
   \frac 1n \E \sqb{d_{\Sigma_X}(W, \Wh)} = {\frac 1 n} \sum_{i=1}^n \lambda_i D_i\,.
\end{align*}   
Minimizing $\frac 1 n \sum_i \lambda_i D_i$ subject to the rate constraint via Lagrange multipliers gives the (reverse) waterfilling optimum, summarized by the parametric curve in \eqref{eq:wf_cruve}.
\end{proof}

Now, as we discussed above, $\Sigma_X$ describes distribution of activations and, practically
speaking, is usually unavailable to decoder. Indeed, even if the eigenvalues $\Lambda$ were known, optimal waterfilling requires knowledge of the
eigenbasis $U$, since the decoder $g_{\Sigma_X}$ essentially computes $U\cdot \widehat{U^\top W}$,
where $\widehat{U^\top W}$ is the closest codeword to $U^\top W$. Communicating $U \in
\mathcal{O}(n)$ to the decoder is expensive: it requires approximating $\Theta(n^2)$ real
parameters, which will consume 
much larger than $\Theta(n)$ bits allocated for communicating $W$ itself.\footnote{In practice, GPU computes not a single inner product
but $W^\top X$ for $W \in
\mathbb{R}^{n\times a}$ being a matrix with $a\asymp n$. Thus, some of the cost of sending $U$
ammortizes over $a$, but still makes it a highly suboptimal choice.} Thus, the tradeoff in \eqref{eq:wf_cruve} is unattainable via a na\"ive ``send $W$ and $\Sigma_X$" scheme when decoder lacks $\Sigma_X$.

To capture the limitation above, we need to assume that $\Sigma_X$ is available at \emph{encoding}
time, but unavailable at \emph{decoding} time (since the deployed dequantizer is fixed, possible
even in hardware). So, formally we define a universal $(n,R)$ quantization scheme as a pair
\begin{equation}\label{eq:enc_dec_def_intro}
    f:\R^{n}\times \mathbb{S}_+^{n}\to [2^{nR}],
    \qquad
    g:[2^{nR}]\to \R^{n}\,,\tag{$f,g$}
\end{equation}
where the encoder takes $\Sigma_X$ as an input, while the decoder
has no access to $\Sigma_X$ and outputs $\Wh = g(f(W, \Sigma_X))$. Again, the image of $g$ is
called the codebook $\CCC=\mathrm{im}\, g$, which is independent of $\Sigma_X$.

The goal is to design a pair $(f,g)$ such that $D(\CCC,\Sigma_X)$ were low simultaneously for all
$\Sigma_X$. Our first main result proves existence of a universal codebook $\CCC$ with an explicit
guarantee on the achieved distortion. To define that guarantee, again let $\Sigma_X \in
\mathbb{S}_+^n$ with $\trace(\Sigma_X) = n$ and spectrum $\lambda = (\lambda_1,\dots,\lambda_n)$.
The \emph{random-coding rate-distortion function} is a parametric curve (with parameter $T > 0$)
given by
\begin{equation}\label{eq:rd_rc_G_intro}
     \Drc(\lambda, T) = \frac 1 n \sum_{i=1}^n \frac{\lambda_i}{1+\lambda_i T}\qquad \Rrc(\lambda, T) = \frac 1 {2 n}\sum_{i=1}^n \log(1+\lambda_i T)\,.\tag{RDRC}
\end{equation}
Denote 
\begin{equation}\label{eq:rd_d_G_intro}
    \DDrc(\lambda, R) \eqcolon \Drc(\lambda, \Trc(\lambda, R))\,,\quad \text{where } \Rrc(\lambda, \Trc(\lambda, R)) = R\,.\tag{$\DDrc$}
\end{equation}

\begin{theorem}[Main Result I: Universal Quantization Scheme for Gaussian Input]\label{thm:union} Fix any constants $\Rstar, \eps, \eta, B > 0$. There exists an encoder $f: \R^n \times \mathbb{S}_{+}^n \times [0,1] \to \sqb{2^{nR}}$, a decoder $g: \sqb{2^{nR}} \times [0,1] \to \R^n$ with $R \leq \Rstar + \eps$, and a (shared) random variable $S\in [0,1]$ with the following property. 
    For $W\sim\N(0,I_n)$ and $\Sigma_X\in \mathbb{S}^n_+$ we set $\Wh(\Sigma_X) = g\paren{f(W,
    \Sigma_X, S), S}$. Then, for sufficiently large $n \geq n_0 = n_0(\eps,\eta, \Rstar, B)$, we
    have with probability at least $1-\exp\paren{-n^B}$ over $S\sim \mathrm{Unif}[0,1]$ that 
    $$
    {\frac 1 n}\E_W \sqb{d_{\Sigma_X}(W,\Wh(\Sigma_X))} \le \DDrc(\spec(\Sigma_X), \Rstar)+ \eta
    $$
    simultaneously for all $\Sigma_X$ with $\trace \Sigma_X=n$.
\end{theorem}

Theorem~\ref{thm:union} uses $S$ to generate an entire codebook $\CCC$ and shows that with high
probability it achieves $\approx\DDrc(\Sigma_X, R)$ distortion. Of course, by fixing a value of
$S$ it implies \emph{existence} of a single codebook $\CCC$ with the same property. Formally, we
have a corollary.

\begin{corollary}\label{cor:existence} Under the same assumptions as in Thm~\ref{thm:union}, there exists
a codebook $\CCC$ of size  $\log_2 |\CCC| \le n( \Rstar + \eps)$ and an encoder-decoder pair $f,g$
(see Eq.~\eqref{eq:enc_dec_def_intro}) with the following property. Set $\Wh(\Sigma_X) = g(f(W,
\Sigma_X))$. Then, simultaneously for all $\Sigma_X$ we have 
$$
 \frac1n \E_{W} \sqb{d_{\Sigma_X}(W, \Wh(\Sigma_X))} \le \DDrc(\spec(\Sigma_X), \Rstar) + \eta
 {\frac{\tr \Sigma_X } n} \,.
$$
\end{corollary}

The description of the scheme and intuition are in Sec.~\ref{sec:ProofSketch} and the full proof
in Appendix~\ref{sec:rc_gaussian_W}. While the results are presented for isotropic Gaussian $W$,
our proof proceeds by showing a more general result for a fixed non-random $W$ (see
Sec.~\ref{subsec:rc_general_W}). Then, the result for Gaussian $W$ is obtained by using 
concentration of measure (Sec.~\ref{sec:rc_gaussian_W}).

\medskip
 What Theorem~\ref{thm:union} shows is that, roughly speaking, there exists a universal codebook of rate
$R$ which achieves distortion $\DDrc(\spec(\Sigma_X),R)$ simultaneously for all $\Sigma_X$. 
A natural question is: \emph{how far is $\DDrc(\spec(\Sigma_X), R)$ from the oracle waterfilling
benchmark $\DDwf(\spec(\Sigma_X),R)$?} 

To make the comparison easier to interpret, we will phrase it in terms of rate overhead (of our
codebook) compared to $\Sigma_X$-fine-tuned optimal codebook. Specifically, for a fixed spectrum $\lambda = \diag(\lambda_1, \dots, \lambda_n) \succeq 0$ and a distortion
level $\Dstar$, let $\RRwf(\lambda, \Dstar) = \Rwf(\lambda, t)$ be the minimum oracle rate
achieving $\Dwf(\lambda, t) = \Dstar$ (see Eq.~\eqref{eq:wf_cruve}), and let $\RRrc(\lambda,
\Dstar) = \Rrc(\lambda, T)$ be the minimum random-coding rate achieving $\Drc(\lambda, T) =
\Dstar$ (see Eq.~\eqref{eq:rd_rc_G_intro}). The difference $\RRrc(\lambda, \Dstar) -
\RRwf(\lambda, \Dstar)$ is the rate overhead incurred by using a universal decoder $g$ that is
agnostic to the covariance matrix $\Sigma_X$ with $\spec(\Sigma_X) = \lambda$. 


Our second main result in Theorem~\ref{thm:worst_case} shows that this overhead is uniformly bounded by 0.11 bit.  Specifically, we derive precise expressions for $\RRrc(\lambda, \Dstar) - \RRwf(\lambda, \Dstar)$ that depend on $\lambda = \spec(\Sigma_X)$ and $\Dstar$ and prove that the maximum gaps occur at the spectra with at most $2$ distinct eigenvalues and vanishing distortions (see Fig.~\ref{fig:worst_gap} for the worst-case rate gap found at each $R = \RRrc(\lambda, \Dstar)$). See Sec.~\ref{sec:worst_case} for the full proof.

\begin{theorem}[Main Result II: Worst-Case Rate Gap to Oracle Setting]\label{thm:worst_case}
    $$
    \sup_{\Dstar \in (0,1)}\sup_{\Lambda} \sset{ \RRrc(\Lambda, \Dstar) - \RRwf(\Lambda, \Dstar)} \leq 0.11\,,
    $$
    where the supremum is over $\Lambda = \diag(\lambda_1,\dots,\lambda_n)\succeq 0$ with $\trace(\Lambda) = n$.

\end{theorem}


\begin{figure}
    \centering
    \includegraphics[width=0.8\linewidth]{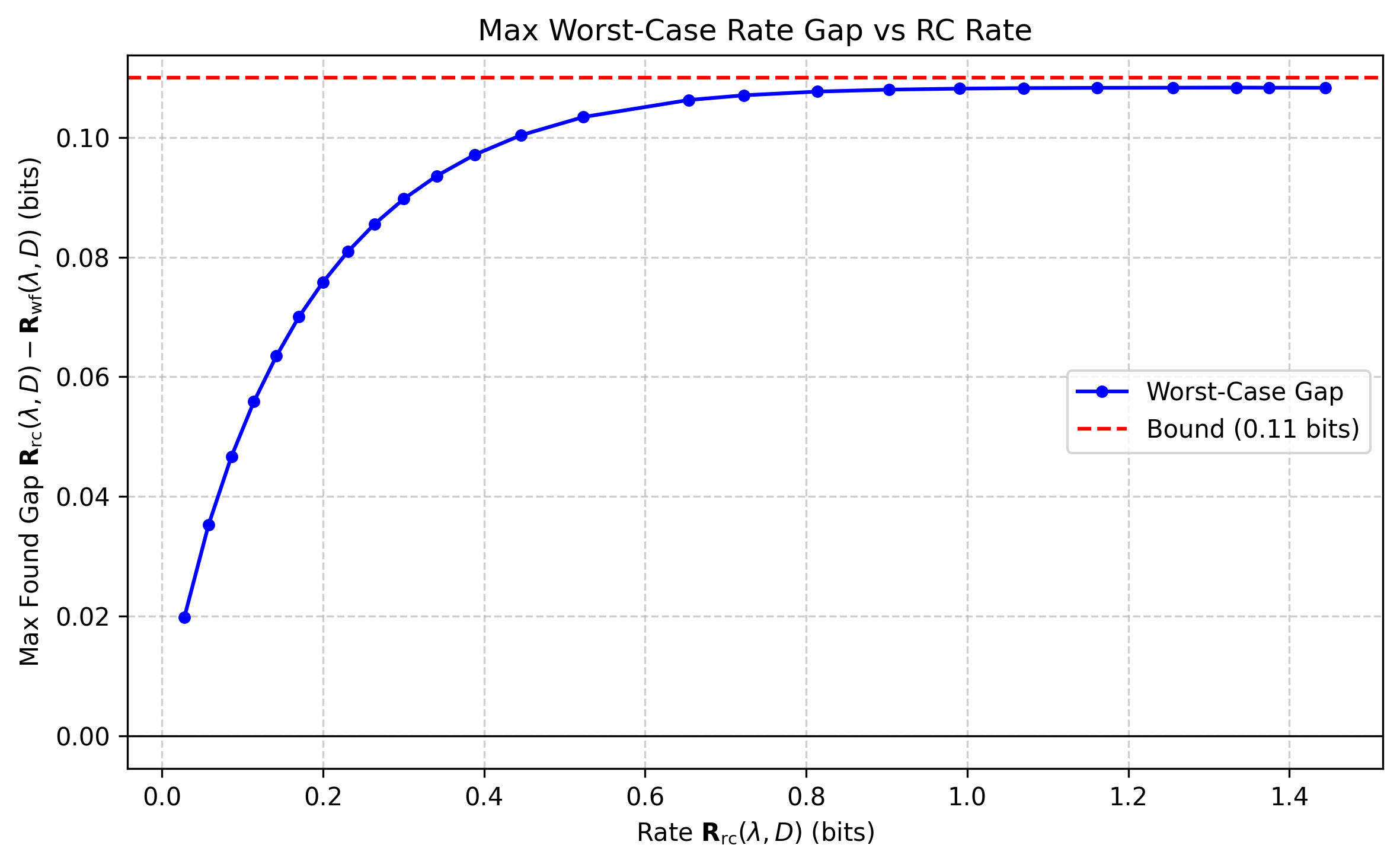}
    \vspace{-3mm}
    \caption{Maximum rate gap found at each rate $R$.}
    \vspace{-5mm}
    \label{fig:worst_gap}
\end{figure}

An intuitive way to see why the gap is bounded is to consider the most extreme of $\trace
\Sigma_X=n$ cases: the identity and the rank-1 case, for which $\RRrc(\lambda, \Dstar) =
\RRwf(\lambda, \Dstar)$ by a simple computation.
In fact, the gap vanishes for all matrices with semi-flat spectra $\lambda = (n/m,\dots,n/m,
0,\dots,0), m\leq n$, a phenomenon we further discuss in Sec.~\ref{subsec:geom_int}.

\paragraph{Discussion and Open Questions.}
Our results prove the \emph{existence} of a universal codebook whose performance is uniformly within a constant rate gap of the $\Sigma_X$-aware (waterfilling) oracle benchmark, showing that universality is not, by itself, an information-theoretic bottleneck. Our random-coding-based proof is nonconstructive and does not yield an explicit codebook design or efficient encoder/decoder pair. Designing explicit and computationally efficient constructions is the most immediate open direction. 

Often, low complexity quantizers are based on lattices. Indeed, for any fixed $\Sigma_X$ a lattice randomly drawn from the natural Haar-Siegel measure will be a good quantizer  with probability $1-e^{-\Omega(n)}$. This follows from the results of~\cite{orw22} that show that the covering radius of a random lattice is typically near-optimal with respect to any (fixed) norm. However, \emph{no lattice can be near-optimal simultaneously for all $\Sigma_X$}, and the reason is simple: for any fixed lattice $L\subset\RR^n$ there is a rotation $U^\top L$ that aligns its directions with the natural basis. Consequently, for this rotation $U^\top L=\prod_{i=1}^n (\alpha_i\ZZ)$ for some $\alpha_1,\ldots,\alpha_n$. The integer lattice is a bad quantizer, and therefore any lattice quantizer must lose at least $\frac{1}{2}\log\frac{2\pi e}{12}\approx 0.254$ bits with respect to the waterfilling benchmark for some $\Sigma_X$ (see~\cite{ordentlich2026high} for more details). It therefore follows that $\Sigma_X$-universal near optimal schemes cannot solely rely on lattice quantizers.

One example of a lattice-based algorithm that provides a practical solution to the problem studied here is the GPTQ algorithm~\cite{frantar2023gptq} (with appropriate shaping/entropy coding~\cite{ordentlich2026high}). Its high-rate gap to the waterfilling benchmark, for particular $\Sigma_X$, is $\frac{1}{2}\log\frac{2\pi e}{12}+\frac{1}{2}\log(\mathrm{AM-GM}(\Sigma_X))$ where $\mathrm{AM-GM}(\Sigma_X)$ is the ratio between arithmetic-mean and geometric-mean of the squared diagonal elements in the Cholesky decomposition of $\Sigma_X$, which can be unbounded in general~\cite{ordentlich2026high}. When $W$ is a matrix consisting of $a\gg 1$ rows (rather than the vector case considered here), some of this gap can be reduced by sending to the decoder $o(na)$ bits of information on $\Sigma_X$ (which has negligible effect on the rate). One such example is the WaterSIC quantization scheme~\cite{lifar2026watersic}. 



\paragraph{Related literature.}

The problem of vector quantization is classical~\cite{gersho2012vector}, and its asymptotic
behavior (for iid sources and additive distortion) is given by the famous rate-distortion
formulas, e.g.~\cite[Part V]{PWbook24}. A recent wave of interest, however, focuses on computing 
quantized inner-product and matrix multiplication.

\textit{On the practical side,} the pioneering work of \cite{dettmers2022gpt3} demonstrated that
substantial compression is possible via scaling plus uniform rounding (INT8 weight quantization),
thus establishing the field of post-training quantization (PTQ). Notable PTQ works include 
SmoothQuant \cite{xiao2024smoothquant}, which introduced calibration-based methods, i.e. those
which depend on statistics of activations $X$ via $\Sigma_X$.
GPTQ~\cite{frantar2023gptq} and LDLQ~\cite{chee2024quip}, which are equivalent, simultaneously introduced
an algorithm for $\Sigma_X$-dependent quantization. Going beyond simple
integer-rounding,~\cite{tseng2024quipsharp} consider lattice, ~\cite{savkin2025nestquant} consider
nested lattice and~\cite{tseng2025qtip} consider trellis quantization methods, respectively. Going
beyond simple quadratic losses, are~\cite{tseng2025yaqa} and~\cite{badri2023hqq}. 
We note that
\cite{tseng2024quipsharp} also reintroduced random Hadamard transform (RHT) as a way of mitigating outliers,
following earlier usage in quantization of gradients, and classically. See
\cite{ashkboos2024quarot,liu2025spinquant,chen2025wush} for other applications of RHT in PTQ.

\textit{On the theoretical side,} the work~\cite{ordentlich2025optimal} established fundamental
limits of quantized matrix multiplication by leveraging nested lattice quantization. The GPTQ/LDLQ
algorithm was understood as Babai's nearest-plane algorithm applied after a Cholesky factorization
of $\Sigma_X$ in \cite{chen2025geometry,birnick2025lattice}, and as a successive interference
cancellation (SIC) algorithm in~\cite{ordentlich2026high}, thus connecting weight-only
quantization to lattice decoding and
approximate closest vector problems \cite{conway1982fast,babai86}.
Authors of \cite{ordentlich2026high} developed theoretical
high-rate analysis of GPTQ, showed it can be arbitrarily far from waterfilling and proposed
an improved algorithm which provably matches waterfilling to within 0.255 bit in the
high-rate regime. 

The question considered here ($\Sigma_X$-oblivious quantization) in the special case of diagonal
$\Sigma_X$ falls under the umbrella of the 
 \textit{compression with distortion as side-information} proposed in~\cite{MWZ08}. 
This viewpoint connects modern LLM quantization to a long tradition in lattice decoding and
approximate closest vector problems \cite{conway1982fast,babai86}.

\section{Technical Overview}
\label{sec:ProofSketch}

As in any source-coding problem, once the codebook $\CCC=\{c_1,\ldots,c_{M}\}\subset\RR^n$ is fixed, the optimal encoder computes $i^*=i^*(W)=\argmin_{i\in[M]} d_{\Sigma_X}(W,c_i)$, and sends the index $i^*$ to the decoder, which in turn outputs $\Wh=c_{i^*}$. Thus, the distortion of the codebook $\CCC$ is
\begin{align}
n D(\CCC,\Sigma_X)=\EE\left[\min_{i\in[M]}d_{\Sigma_X}(W,c_i)\right],    
\end{align}
where the expectation is with respect to $W\sim\m{N}(0,I_n)$. Even if we could design $\CCC$ based on $\Sigma_X$, the distortion $D(\CCC,\Sigma_X)$ must satisfy (see Sec.~\ref{sec:main_results}) the waterfilling lower bound 
$$
D(\CCC, \Sigma_X) \geq \DDwf(R, \Sigma_X)\,;
$$
the lower bound is also asymptotically achievable in the limit of large $n$ (see Proposition~\ref{prop:oracle_wf}).
The challenge is to find a single codebook $\CCC$ with $M=2^{nR}$ codewords in $\RR^n$ that attains small $D(\CCC, \Sigma_X) - \DDwf(R, \Sigma_X)$ 
simultaneously for all $\Sigma_X$. 

As is standard, we prove the existence of such a codebook $\CCC$ by drawing a random code with $M=2^{nR}$ iid codewords from a distribution $P_{\Wh}$. We show in Theorem~\ref{thm:union} that for appropriate choice of $P_{\Wh}$ we have that 
\begin{align}\label{eq:ideas_prob_codebook}
\Pr_{\CCC}\left[\sup_{\Sigma_X} \left(D(\CCC,\Sigma_X)-\DDrc(\spec(\Sigma_X), R-\eps)\right)<\eta\right] \geq 1 - \exp\paren{-\poly(n)}    
\end{align}
holds for any $\eta,\eps>0$ and $n$ large enough, where $\spec(\Sigma_X)$ is the vector of the eigenvalues of $\Sigma_X$ and $\DDrc$ is defined in Eq.~\eqref{eq:rd_d_G_intro}.  Consequently there must exist a fixed rate-$R$ codebook $\CCC$ with
\begin{align}
D(\CCC,\Sigma_X)\leq \DDrc(\spec(\Sigma_X), R-\eps)+\eta,~~~\forall\Sigma_X\in\mathbb{S}_+^n \text{ with }\trace(\Sigma_X)=n.    
\end{align}

\subsection{Codebook Distribution}\label{subsec:ideas_codebook}
How should we choose $P_{\Wh}$? For a given $\Sigma_X$ with spectral decomposition $\Sigma_X = U \Lambda U^\top$, the optimal $P_{\Wh}$ follows from the waterfilling solution. Specifically, for water-level $1/t$ chosen so that $\Rwf$ defined in~\eqref{eq:wf_cruve} equals $R$, the optimal distribution is 
\begin{align}
P^*_{\Wh}(\Sigma_X,R)&=P^*_{\Wh}(\Sigma_X,t)=\m{N}(0,U~\Gamma(\Lambda,t)~U^\top),\nonumber\\
&\text{where}~~\Gamma(\Lambda,t)=\diag\left(\max\left\{1-\frac{t}{\lambda_1},0\right\},\cdots,\max\left\{1-\frac{t}{\lambda_n},0\right\}\right).   \label{eq:wfOptDist}
\end{align}
We need to choose a single $P_{\Wh}$ that ``works well'' for all $\Sigma_X$. Since there is no preference to any $U\in\m{O}_n$, it makes sense to take an isotropic Gaussian $P_{\Wh}$. Observing that for any fixed $\Sigma_X$ and $t>0$ the covariance matrix for $P^*_{\Wh}(\Sigma_X,t)$ satisfies $U~\Gamma(\Lambda,t)~U^\top \preceq I_n$, we will take 
\begin{align}
P_{\Wh}(\tau)=\m{N}(0,\tau^2 I_n),    
\end{align}
for some $0<\tau<1$. 

An appealing feature of the isotropic Gaussian distribution is that drawing $M$ iid vectors from $\m{N}(0,\tau^2 I_n)$ is equivalent to first drawing them iid from $\m{N}(0,I_n)$ and then scaling all of them by $\tau$. The consequence of this simple fact is that while we cannot perfectly match our codebook distribution to $P^*_{\Wh}(\Sigma_X,R)$,  the flexibility in the choice of $\tau = \tau(\Sigma_X, R)$ allows for a better match. 
Consequently, we draw the $M$ codewords of $\CCC$ from the $\N(0,I_n)$ distribution. The encoder, that knows $\Sigma_X$, computes $\tau(\Sigma_X,R)$ that provides the smallest expected distortion, and sends a description of $\tau$ to the decoder.\footnote{As we will see below, some further gain can be attained by allowing $\tau$ to also depend on the source realization $w\in\RR^n$.} With this procedure, the effective codebook $\tilde{\CCC}=\tau\CCC$ is drawn from $P_{\Wh}=\m{N}(0,\tau^2(\Sigma_X,R) I_n)$. The encoder then finds $i^*=\argmin_{i\in[M]}d_{\Sigma_X}(W,\tilde{c}_i)=\argmin_{i\in[M]}d_{\Sigma_X}(W,\tau c_i)$ and sends $i^*$ as well as a high-resolution description of $\tau$ in bits to the decoder. 




\subsubsection{Geometric Intuition}\label{subsec:geom_int}

To get some intuition to why the ``universal'' codebook distribution $P_{\Wh}=\m{N}(0,\tau^2(\Sigma_X,R)I_n)$ works well simultaneously for all $\Sigma_X$, let us restrict attention to the family of covariance matrices with semi-flat spectrum. In particular, for $m\leq n$ let
\begin{align}
\m{S}_m^n=\left\{\Sigma_X=U\Lambda U^\top~:~U\in\m{O}_n,~\lambda_1=\cdots=\lambda_m=\frac{n}{m},~\lambda_{m+1}=\cdots=\lambda_n=0 \right\},    
\end{align}
be the collection of PSD matrices with $m\leq n$ equal and non-zero singular values and $m-n$ zero singular values. From~\eqref{eq:wfOptDist} we see that the $\Sigma_X$-matched optimal codebook distribution is of the form
\begin{align}
P^*_{\Wh}(\Sigma_X,R)=\m{N}\left(0,U\cdot\diag\left(\tau^2,\cdots,\tau^2,0,\cdots,0\right)U^\top \right)    
\end{align}
for some $\tau>0$. Thus the optimal procedure for random coding is to draw iid Gaussian codewords within the subspace $U_{[m]}$ spanned by the first $m$ singular vectors. Our universal distribution, on the other hand, draws isotropic iid Gaussian codewords, and is hence very far from the optimal distribution. However, since the encoder searches for the nearest codeword under the $d_{\Sigma_X}$ metric, it effectively projects both $W$ and the codebook $\CCC$ to $U_{[m]}$ and finds the nearest codeword in $\ell_2$-metric within this subspace. It therefore follows that what dictates performance of a random code $\CCC$ under $d_{\Sigma_X}$ metric (for $\Sigma_X$ with semi-flat spectrum) is the distribution of $U^\top_{[m]}c$, where $c\sim P_{\Wh}$. Thus, our universal $P_{\Wh}=\m{N}(0,\tau^2 I_n)$ is simultaneously optimal for all semi-flat $\Sigma_X$ (with all possible $m\in[n]$), provided that we judiciously choose $\tau=\tau(\Sigma_X,R)$. Inspection of our $\Sigma_X$-universal rate-distortion tradeoff~\eqref{eq:rd_rc_G_intro} shows that indeed $$\DDrc(\spec(\Sigma_X),R)=\DDwf(\spec(\Sigma_X),R),$$ for all semi-flat $\Sigma_X$. Whenever the spectrum of $\Sigma_X$ is not semi-flat the distribution of $U^\top c$ for $c\sim P_{\Wh}$ does not match that of $U^\top c$ under the optimal $P^*_{\Wh}(\Sigma_X,R)$, and consequently in these cases our $\Sigma_X$-universal rate-distortion tradeoff~\eqref{eq:rd_rc_G_intro} is worse than the waterfilling rate-distortion tradeoff. Nevertheless, it turns out that the loss for this mismatch is at most $0.11$ bits, as shown in Theorem~\ref{thm:worst_case}.

\subsection{Sketch of Proof}\label{subsec:ideas_proof_sketch}

After we obtained intuition for the choice of using a random iid isotropic Gaussian codebook, with scale $\tau$ determined by the encoder, we move on to giving an overview of the proof of Theorem~\ref{thm:union}. The detailed rigorous proof is given in Appendix~\ref{sec:rc_gaussian_W}.

\medskip

Let us first fix $\Sigma_X=U\Lambda U^\top$, and analyze the performance of a codebook $\CCC=\{c_1,\ldots,c_M\}$, $M=2^{nR}$, with $c_i\stackrel{iid}{\sim}\m{N}(0,I_n)$. We will show that $\CCC$ is ``good'' for $\Sigma_X$ with probability $1-\exp(-e^{\Omega(n)})$, and from this we will deduce that a random $\CCC$ is ``good'' for all $\Sigma_X$ by a covering argument. 
\medskip

The codebook $\CCC$ can describe a fixed $w\in\RR^n$ with distortion $\leq D$ if at least one of its $\tau$-scaled codewords is inside the region $w+\sqrt{D}\m{B}_{\Sigma_X}$, where
\begin{align}
\m{B}_{\Sigma_X}=\left\{e\in\RR^n~:e^\top \Sigma_X e \leq n \right\}.
\end{align}
Thus, the key to analyzing the tradeoff between rate and distortion is understanding how the success probability of a single codeword behaves as a function of $D$. Since the codewords are $\m{N}(0,I_n)$, their distribution is invariant to rotation, and therefore the success probability of a single codeword is 
\begin{align}
p_\suc(w,\tau,\Sigma_X,D)&=\phi_{\tau^2}\left(w+\sqrt{D}\m{B}_{\Sigma_X}\right)\nonumber\\
&=\phi_{\tau^2}\left(U^\top w+\sqrt{D}\m{B}_{\Lambda}\right)=p_\suc(U^\top w,\tau,\Lambda,D),    
\label{eq:pnAsGaussianMeasure}
\end{align}
where $\phi_{\tau^2}$ denotes the probability distribution for $\m{N}(0,\tau^2 I_n)$.
When $p_{\suc}<2^{-n(R+\eps)}$, it is very unlikely to find a codeword in
$w+\sqrt{D}\m{B}_{\Sigma_X}$, and on the other hand, when $p_\suc>2^{-n(R-\eps)}$ we are very
likely to find a codeword in $w+\sqrt{D}\m{B}_{\Sigma_X}$. Thus, what we are looking for is the
critical $D$ for which $-\frac{1}{n}\log p_{\suc}(U^\top w,\tau,\Lambda,D)\approx R$. Since the
codebook's scale $\tau$ needs to be sent from the encoder to the decoder anyway, we may let it
depend not only on $\Lambda$ but also on $U^\top w$. Therefore, given $U^\top w,\Lambda$ and $R$
we choose $\tau=\tau(U^\top w,\Lambda,R)$ for which the critical $D$ is small. A tedious but
straightforward calculation shows that the optimal choice is
\begin{align}
\tau=\tau(U^\top w,\Lambda,R)=\paren{T \sum_{j} \frac{(U^\top w)_j^2 \lambda_j^2}{(1+\lambda_j T)^2}}^{1/2}\paren{\sum_j \frac{\lambda_j}{1+\lambda_j T}}^{-1/2}, 
\label{eq:tauSketch}
\end{align}
where $T=T(\Lambda,R)$ is such that $\Rrc(\Lambda,T)=R$, and $\Rrc(\Lambda,T)$ is defined in~\eqref{eq:rd_rc_G_intro}. Let
\begin{align}
\Drc(U^\top w)=\Drc(\Lambda, R,U^\top w) = \frac 1n \sum_{i=1}^n \frac{(U^\top w)_i^2 \lambda_i}{1+\lambda_i T},   
\end{align}
where here as well $T=T(\Lambda,R)$. In Lemma~\ref{lem:codeword_success} we prove that 
\begin{align}
-\log p_\suc(U^\top w,\tau,\Lambda,\Drc(U^\top w))=-\log \phi_{\tau^2}\left(U^\top w+\sqrt{\Drc(U^\top w)}\m{B}_{\Lambda}\right)\approx nR.   
\label{eq:critD}
\end{align}
Thus, $\Drc(U^\top w)$ is the critical distortion for fixed $w\in\RR^n$ and $\Sigma_X$. Note that for $W\sim\m{N}(0,I_n)$ we have
\begin{align}
\EE_W\left[\Drc(U^\top W) \right] =  \frac 1n \sum_{i=1}^n \frac{\lambda_i}{1+\lambda_i T} = \Drc(\Lambda, R),
\end{align}
where $\Drc(\Lambda, R)$ is the distortion in~\eqref{eq:rd_d_G_intro}.

The proof of Lemma~\ref{lem:codeword_success} uses standard large deviation techniques, but note that what we really need for the analysis is a quantitative lower bound on $p_\suc$, and this requires some more work beyond Chernoff bound. The choice of $\tau$ from~\eqref{eq:tauSketch} is the one that minimizes the large deviations exponent, whereas the parameter $T$ in our rate-distortion tradeoff~\eqref{eq:rd_rc_G_intro} is just a rescaling of the parameter $t$ in Chernoff's bound $\Pr(X>D)\leq e^{-t D}\EE[e^{tX}]$.

\medskip

By~\eqref{eq:critD}, for $\eta>0$ we have that $p_\suc(U^\top w,\tau, \Lambda,\Drc(U^\top w)+\eta/2)\geq 2^{-n(R-\eps_\eta)}$ for some $\eps_\eta>0$ (and $n$ large enough). Thus, $\forall w\in\RR^n$ 
\begin{align}
P_{\mathrm{failure}}(w)&=\Pr_{\CCC}\left(\frac{1}{n}\min_{i\in[M]}d_{\Sigma_X}(w,\tau c_i)> \Drc(U^\top w)+\eta/2\right)\nonumber\\
&=1-\left(1-p_\suc\left(U^\top w,\tau,\Lambda,\Drc(U^\top w)+\eta/2\right)\right)^M \leq \exp (-2^{n\eps}), 
\label{eq:Pfailure}
\end{align}
for some $\eps>0$. For a \emph{fixed} code $\CCC$ define
\begin{align}
\m{E}_{\mathrm{failure}}(\CCC)=\left\{w\in\RR^n~:~\frac{1}{n}\min_{i\in[M]}d_{\Sigma_X}(w,\tau c_i) > \Drc(U^\top w)+\eta/2 \right\}.
\end{align}
Assuming we can always encode $w$ to $0$ (e.g., by setting $\tau=0$ if needed), for any $W\sim P_W$ we have
\begin{align}
\frac{1}{n} D(\CCC,\Lambda)&=\EE\left[\min_{i\in[M]}d_{\Sigma_X}(W,\tau c_i) \right]\leq \EE\left[\Drc(U^\top W)+\frac{\eta}{2}+\Ind\{W\in \m{E}_{\mathrm{failure}}(\CCC)\}d_{\Sigma_X}(W,0) \right]\nonumber\\
&\leq \EE\left[\Drc(U^\top W)\right]+\frac{\eta}{2}+\sqrt{P_W[\m{E}_{\mathrm{failure}}(\CCC)]}\cdot\sqrt{\EE[d^2_{\Sigma_X}(W,0)]},
\end{align}
where the last inequality is Cauchy-Schwarz. If $\EE_W[d^2_{\Sigma_X}(W,0)] \leq \poly(n)$ under $W\sim P_W$,\footnote{e.g., $\EE_W[d^2_{\Sigma_X}(W,0)] \leq O(n^2)$ for $W\sim \N(0,I_n)$} it follows that whenever $P_W\sqb{\m{E}_{\mathrm{failure}}(\CCC)}$ is sufficiently small, say $\leq e^{-n}$, the codebook $\CCC$ attains the desired $$\frac{1}{n} D(\CCC,\Lambda)\leq \EE\left[\Drc(U^\top W)\right]+\eta\,.$$
$P_W\sqb{\m{E}_{\mathrm{failure}}(\CCC)}$ can indeed be bounded: using~\eqref{eq:Pfailure} and Markov's inequality we obtain that this holds for the vast majority of codebooks:
\begin{align}
\Pr_\CCC\left(P_W[\m{E}_{\mathrm{failure}}(\CCC)]>e^{-n} \right) &\leq e^n\EE_{\CCC}\left[P_W[\m{E}_{\mathrm{failure}}(\CCC)]\right]= e^n \EE_{\CCC,W}[P_{\mathrm{failure}}(W)]\nonumber\\
&\leq \exp (-2^{n\eps}+n).
\end{align}
From this we conclude that for fixed $P_W$ and any fixed $\Sigma_X$
\begin{align}
\Pr_{\CCC}\left(\left[\frac{1}{n} D(\CCC,\Lambda)- \EE_W\left[\Drc(U^\top W)\right]>\eta \right]\right)\leq \exp(-2^{n\eps'}).   
\end{align}
Specializing this to $P_W=\m{N}(0,I_n)$ we obtain
\begin{align}
\Pr_{\CCC}\left(\left[\frac{1}{n} D(\CCC,\Lambda)- \Drc(\Lambda, R)\right]>\eta\right)\leq \exp(-2^{n\eps'}).   
\end{align}

We have therefore obtained that for any fixed $\Sigma_X$ the probability of drawing a ``bad'' codebook with $\left[\frac{1}{n} D(\CCC,\Lambda)- \Drc(\Lambda, R)\right]>\eta$ is double-exponentially small. From here, it is clear how to prove that a randomly drawn $\CCC$ will have 
\begin{align}
\Pr_{\CCC}\left(
\sup_{\Sigma_X}\left[\frac{1}{n} D(\CCC,\Lambda_X)- \Drc(\Lambda, R)\right]>\eta\right)<\exp(-2^{n\eps''}).
\end{align}
All we need is to find a dense cover of PSD matrices in $R^{n\times n}$ with trace $n$ whose size is $\exp(-2^{o(n)})$. In particular, we need to find a collection of $N=\exp(-2^{o(n)})$ PSD matrices with the property that for any valid PSD matrix $\Sigma_X$ there exists a matrix $\Sigma_i$, $i\in[N]$ such that:
\begin{equation}\label{eq:ideas_perturb}
\Drc(\spec(\Sigma_X), R)\approx \Drc(\spec(\Sigma_i), R)\qquad\text{and}\qquad \frac{1}{n}D(\CCC,\Sigma_X)\approx \frac{1}{n}D(\CCC,\Sigma_i)\,.
\end{equation}
Since $N$ is allowed to be so large, these two constraints can be met to arbitrary resolution (though the proof for this requires a lot of bookkeeping and is somewhat technical, see Sec.~\ref{subsec:proof_of_union}).

\medskip

We end this overview with listing some of the technical issues that our sketch of proof above ignored, and briefly mention how addressing them, as we of course do in the actual proofs, affects the results.

\textbf{$\tau$ quantization.} While our sketch assumed that $\tau$ can be conveyed to the decoder in perfect resolution, in reality, some of our $nR$ bits budget is allocated to the description of $\tau$. In order to compress $\tau$ that depends on $U^\top W$, we require a high probability bound on $\|U^\top W\|_{\infty}$. This norm constraint is also needed for the proof of our large deviations result (Lemma~\ref{lem:codeword_success}).

\textbf{Norm bound for codewords.} The perturbation argument~\eqref{eq:ideas_perturb} for the dense cover must account for quantization and change in $\tau$. Since these errors are multiplied by codewords from $\CCC$, we require a uniform norm bound $\|c_i\|_2<n^B$ for some $B>1$ (say $B=10$) to make their contribution to end-to-end distortion negligible.
While the probability that this occurs is overwhelmingly large, it is only $1-\exp(-\poly(n))$ rather than double exponential. For this reason the probability of drawing a $\Sigma_X$-universal   codebook is only $1-\exp(-\poly(n))$ (see Eq.~\eqref{eq:ideas_prob_codebook}) rather than the $1-\exp(-2^{\Omega(n)})$ that our sketch of proof gives. Because event $\{\|c_i\|_2<n^B~~\forall i\in[M]\}$ is independent of $\Sigma_X$, it contributes only a single term to the union bound. Consequently, our dense grid is still allowed to be of size $\exp(-2^{o(n)})$.

\textbf{Rate-penalty.} In the overview above, we assumed that if $-\ln p_\suc(D)=R$, then for any $\eta>0$ we have $-\ln p_\suc(D+\eta)=R-\eps_\eta$ for some $\eps_\eta>0$. This is indeed the case whenever $D'(R)$ is finite. However, our results are for the supremum over all $\Sigma_X$ with trace $n$, and this cannot be guaranteed for all such matrices at all rates. Consequently, in Theorem~\ref{thm:union} there is both a rate-penalty $\eps>0$ and a distortion penalty $\eta>0$ (which both can be made arbitrarily small for $n$ large enough), whereas in the sketch above we only had a distortion penalty. 

\section*{Disclosure of LLM Assistance}
The authors used ChatGPT to assist with editing, code generation for Fig.~\ref{fig:worst_gap}, and technical steps (e.g., perturbation analysis and derivative computations) in the proofs of main theorems. The final optimization step after the spectrum reduction in Theorem~\ref{thm:worst_case} was also suggested by ChatGPT. All model-generated code, computations, and proof steps were independently verified and adjusted by the authors. The authors take full responsibility for the correctness of all analytical and numerical results in the paper. 

\bibliography{wquant}
\bibliographystyle{alphaurl} 

\newpage
\appendix

\section{Preliminaries and Notation}\label{sec:prelim}
\paragraph{Notation.}
We write $\Sigma \succeq 0$ to denote that the matrix $\Sigma\in\R^{n\times n}$ is positive semidefinite; we define the set of positive semidefinite matrices as $\mathbb{S}_+^n = \sset{\Sigma\in\R^{n\times n}:\, \Sigma \succeq 0}$. 

For $\Sigma \in \mathbb{S}_{+}^n$ we denote the spectral decomposition as $$\EVD(\Sigma) = U \Lambda U^\top\,,$$
where $U \in \mathcal{O}_n$ is orthogonal and $\Lambda$ is diagonal. 

$\diag(v)$ for $v\in\R^n$ denotes an $n\times n$ diagonal matrix with $\diag(v)_{ii} = v_i$. $\spec(A)$ for $A \in \R^{n\times n}$ denotes the vector of the eigenvalues of $A$. 

For $A \in \R^{n\times a}$ and $\Sigma \in \mathbb{S}_+^n$, denote 
$$
\|A\|_{\Sigma} \triangleq \sqrt{\trace(A^\top \Sigma A)}\,.
$$

\begin{proposition}[Hanson-Wright Concentration Inequality]\label{prop:hanson_wright}
    For $X \sim \N(0, I_n)$ and $A \in \R^{n\times n}$, for every $t\geq 0$,
    $$
    \Pr\sqb{ | X^\top A X - \trace A| > t } \leq 2\exp\paren{-c \min \paren{\frac {t^2} {K^4 \| A\|_F^2}, \frac{t}{K^2 \|A\|_{op}}}}\,,
    $$
    for universal constants $c, K$.
\end{proposition}



\section{Random Coding: Worst-Case $W$}\label{subsec:rc_general_W}

In this section, for fixed $\Sigma_X \in \mathbb{S}_+^n$, we characterize the rate-distortion of quantizing a fixed vector $W \in \R^n$ under a distortion metric $d_{\Sigma_X}$ using a random coding scheme. In particular, for a given $\Sigma_X = U^\top \Lambda U$ with $\Lambda = \diag(\lambda_1,\dots,\lambda_n)$, fixed vector $W \in \R^n$, and constant $\Rstar > 0$, we define a generalized distortion function $\DDrc^{\lambda}(U^\top W, \Rstar)$ and demonstrate that it is achievable in Theorem~\ref{thm:fixed_sigma}.

\paragraph{General Rate-Distortion Function.}  Let $V, \lambda = (\lambda_1,\dots,\lambda_n)^\top \in \R^n$ be such that $\lambda_i \geq 0$ for all $i\in\sqb{n}$ and $\sum_i \lambda_i = n$. We define \emph{general random-coding rate-distortion function} in dimension $n$ to be the following parametric curve for $T > 0$:
\begin{equation}\label{eq:rd_rc}
    \Drc^{\lambda}(V, T) = \frac 1n \sum_{i=1}^n \frac{V_i^2 \lambda_i}{1+\lambda_i T}\qquad\text{and}\qquad \Rrc^{\lambda}(T) = \frac 1{2 n}\sum_{i=1}^n \log(1+\lambda_i T)\,.\tag{RDRC}
\end{equation}
Throughout, $\log$ denotes $\log_2$ and $\ln$ denotes the natural logarithm.
Denote $\Trc^{\lambda}(R)$ be a unique value $T$, s.t. $\Rrc^{\lambda}(T) = R$ (note that $\Trc^\lambda(R)$ is independent of $V$). Let
\begin{equation}\label{eq:rd_d}
    \DDrc^{\lambda}(V, R) \triangleq \Drc^{\lambda}(V, \Trc^\lambda(R))\,.\tag{$\DDrc$}
\end{equation}
We consider a task of quantizing a given vector $W \in \R^n$ under a distortion
function~\eqref{eq:dist_intro}. 

\begin{condition}[Admissible $W, \Sigma_X$]\label{cond:W_Sigma}
    We consider $W \in \R^n, \Sigma_X\in\mathbb{S}^n_{+}$ satisfying 
\begin{enumerate}
    \item $\mathrm{EVD}\paren{\Sigma_X} = U \Lambda U^\top$ for $U \in \mathcal{O}_n$ and $\Lambda = \diag\paren{\lambda_1,\dots,\lambda_n}\succeq 0$ with $\trace(\Lambda) = n$;
    \item $\|U^\top W\|_{\infty} \leq n^{\alpha}$ for a known constant $\alpha$.
\end{enumerate}
\end{condition}

In Theorem~\ref{thm:fixed_sigma} we show that there exists an encoder-decoder pair that, for any constant target rate $\Rstar > 0$ and fixed admissible $W, \Sigma_X$ (Cond.~\ref{cond:W_Sigma}), achieves rate $\Rstar$ and distortion $\DDrc^{U^\top W, \lambda}(\Rstar)$ asymptotically, with high probability over the randomness $S$ shared between encoder and decoder (from which a codebook is generated).\footnote{In the regime $\alpha < 1/4$, see the statement of Thm.~\ref{thm:fixed_sigma}.}

\begin{theorem}[Achievability of Random-Coding Rate-Distortion: Nonasymptotic Guarantee]\label{thm:fixed_sigma}
    Fix any constants $\Rstar, \eps_\circ, \eta_\circ > 0$ and $\alpha > 0$. There exists an
    encoder $f: \R^n \times \mathbb{S}_+^n \times [0,1] \to
    \sqb{2^{nR}}$, a decoder $g: \sqb{2^{nR}} \times [0,1] \to \R^n$ with $R \leq \Rstar + \eps_\circ$ and a (shared) random
    variable $S\in[0,1]$ with the following property. For any fixed $W, \Sigma_X$ (satisfying Cond.~\ref{cond:W_Sigma}) we set $\Wh =
    g(f(W,\Sigma_X,S),S)$. Then for any sufficiently large $n \ge n_0 = n_0 (\eps_\circ,
    \eta_\circ, \Rstar, \alpha)$  and any $\beta > \alpha -
    1/4$ in case $\alpha \geq 1/4$ and $\beta = 0$ otherwise, we have
    $$
    \Pr_{S}\sqb{\frac 1 n d_{\Sigma_X}(W, \Wh) \leq \DDrc^{\lambda}(U^\top W, \Rstar) + n^{2\beta} \eta_\circ} \geq 1 - \exp\paren{-2^{n\eps_\circ(1-cn^{4(\alpha-\beta)-1})}}\,,
    $$
    where $d_{\Sigma_X}(W,\Wh)$ is the distortion function in Eq.~\eqref{eq:dist_intro} and $c = c(\eps_\circ,\eta_\circ,\Rstar,\alpha)$.
\end{theorem}
\begin{remark}
    In case of $\alpha < 1/4$, the distortion bound above simplifies to $\DDrc^{\lambda}(U^\top W, \Rstar) + \eta$ (since $\beta = 0$).
\end{remark}

\subsection{Proof of Theorem~\ref{thm:fixed_sigma}}

\paragraph{Quantization Scheme.} 
    From the shared randomness $S$, $f,g$ generate a Gaussian codebook $\CCC =
    \sset{c_1,\dots,c_{M=2^{n(\Rstar+\eps)}}}$ for $\eps = \eps(\eps_\circ)$ to be chosen
    later. Denote $\Wt = U^\top W$. We define:
    \begin{itemize}
        \item \textbf{Encoder $f$:} Let $T \coloneq \Trc^\lambda(\Rstar)$ and define the scaling parameter $\tau =
	\tau(\Wt, \lambda)$ as 
        \begin{equation}\label{eq:tau}
        \tau = \paren{T \sum_{j} \frac{\Wt_j^2 \lambda_j^2}{(1+\lambda_j T)^2}}^{1/2}\paren{\sum_j \frac{\lambda_j}{1+\lambda_j T}}^{-1/2}\,.\tag{$\tau$ def.}
        \end{equation}
        We set $f$ to be a tuple
        $$
        f(W, \Sigma_X, S) = \paren{\argmin_{i \in \sqb{M}} d_{\Sigma_X}(W, \tau c_i), q(\tau)}\,,
        $$
        where $q(\tau) = \delta \|\Wt\|_{\infty} \lfloor \tau/(\delta \|\Wt\|_{\infty}) \rceil$ is a rounding quantization scheme of precision $\delta$ and recall the distortion function
    $$
    d_{\Sigma_X}(W, C) \coloneq (W - C)^\top \Sigma_X (W - C)\,.
    $$
        \item \textbf{Decoder $g$:}
        $$
        g(i, q(\tau), S) = q(\tau) \cdot c_i\,.
        $$
    \end{itemize}

\paragraph{Rate-Distortion Bound.}  
    In the quantization scheme above, 
    $$
    R = \underbrace{\Rstar + \eps}_{\text{Gauss. codebook}} + \underbrace{\frac1n \log(1/\delta)}_{\tau\text{ quant.}}\,.
    $$
    The rest of the proof is to obtain a high probability bound on the resulting distortion that, given $\Wh = g(f(W,\Sigma_X, S), S) = q(\tau)\cdot c_i$, can be expressed as:
    \begin{align*}
        d_{\Sigma_X}(W,\Wh) = d_{\Sigma_X}(W,q(\tau)c_i) = (W - q(\tau) c_i)^\top \Sigma_X (W - q(\tau) c_i)\,.\label{eq:dist_expand}
    \end{align*}
    In what follows, denote $\Dstar \coloneq \DDrc^{\lambda}(\Wt, \Rstar)$.

    Before giving the proof we state two helpful claims that bound the effect of quantizing $\tau$ in the scheme above. Proofs of Claim~\ref{claim:tau_bound} and \ref{claim:tau_quant_bound} are found in this subsection below.

    \begin{claim}[Bound on $\tau$.]\label{claim:tau_bound}
        The value $\tau$ in Eq.\eqref{eq:tau} satisfies
        $$
        0 \leq \tau \leq \|\Wt\|_{\infty} \leq n^{\alpha}\,.
        $$
    \end{claim}
    
    \begin{claim}[Bound on distortion from $\tau$ quantization]\label{claim:tau_quant_bound} Given $|q(\tau) - \tau| = \delta_\tau \leq \delta n^{\alpha}$, with probability at least $1 - 2^{n(\Rstar+\eps)}\exp(-Ct)$ for a universal constant $C$ and any $t>1/n$,
    $$
    d_{\Sigma_X}(W, q(\tau)c_i) \leq \paren{ \sqrt{d_{\Sigma_X}(W, \tau c_i)} + \delta n^{\alpha} \sqrt{ n(1+t)}}^2\,.
    $$
    \end{claim}

    The main part of the argument is essentially contained in the following Lemma~\ref{lem:codebook_success}. The proof, which also explains the expressions for $\Drc$ and $\Rrc$ is proven in a separate
    section~\ref{subsec:proof_codebook_success} due to its importance. 

    \begin{mylem}{\ref{lem:codebook_success}}\textnormal{\textbf{(Gaussian Book Success)}} Let $\eps,\eta >0$ be constants and $\CCC = \sset{c_1,\dots,c_{M= 2^{n(\Rstar + \eps)}} }$, $c_i \sim_{i.i.d.} \N(0, I_n)$ be a randomly generated Gaussian codebook.

    For admissible $W, \Sigma_X$ (Cond.~\ref{cond:W_Sigma}) with $\alpha < 1/4$, 
        $$
        \Pr_{\CCC}\sqb{ \frac 1 n \min_{i \in \sqb{M}} d_{\Sigma_X}(W, \tau(W,\Sigma_X) \cdot c_i) \leq  \DDrc^{\lambda}(U^\top W, \Rstar)+\eta} \geq 1 - \exp\paren{-2^{n\eps(1-c n^{4\alpha - 1})}}\,,
        $$
        where $c = c(\eps,\eta, \Rstar, \alpha)$ is an explicit constant function, $\DDrc^{\lambda}(U^\top W, \Rstar)$ is defined in Eq.~\eqref{eq:rd_d}, and $\tau(W,\Sigma_X)$ is defined in Eq.~\eqref{eq:tau}.
\end{mylem}

\begin{proof}(of Theorem~\ref{thm:fixed_sigma})
We apply Lemma~\ref{lem:codebook_success} to $W' = n^{-\beta}W$ for $\eta = \eta(\eta_\circ)$ to be chosen later and any constant $\beta > \alpha - 1/4$ in case $\alpha \geq 1/4$ and $\beta = 0$ otherwise. We have $\|U^\top W'\|_{\infty} = n^{-\beta} \|U^\top W\|_{\infty} \leq n^{\alpha - \beta} = o(n^{1/4})$, and therefore,
\begin{align*}
    \Pr_{\CCC}\sqb{ \min_{i \in \sqb{M}} d_{\Sigma_X}(W, \tau c_i) \leq  n \b(\underbrace{\DDrc^{\lambda}(\Wt, \Rstar)}_{\Dstar}+n^{2\beta} \eta\b)} &= \Pr_{\CCC}\sqb{ \min_{i \in \sqb{M}} d_{\Sigma_X}(W', \tau n^{-\beta} c_i) \leq  n \paren{\DDrc^{\lambda}(U^\top W', \Rstar)+\eta}}\\ &\geq 1 - \exp\paren{-2^{n\eps(1-cn^{4(\alpha - \beta) - 1})}}\,.
\end{align*}
From Claim~\ref{claim:tau_bound}, the simple rounding quantizer $q(\tau) = \delta \|\Wt\|_{\infty} \lfloor \tau/(\delta \|\Wt\|_{\infty}) \rceil$ achieves $|q(\tau) - \tau| \leq \delta n^\alpha$, and therefore, by Claim~\ref{claim:tau_quant_bound} and a union bound, with probability at least $1 - \exp\paren{-2^{n\eps(1-cn^{4(\alpha-\beta)-1})}} - 2^{n(\Rstar+\eps)}\exp(-Ct)$, 
\begin{equation}\label{eq:dist_final_unsimpl}
    d_{\Sigma_X}(W,q(\tau) c_i) \leq \paren{ \sqrt{n (\Dstar + n^{2\beta} \eta)} + \delta n^{\alpha} \sqrt{ n(1+t)}}^2\,.
\end{equation}
It remains to simplify the expression in Eq.~\eqref{eq:dist_final_unsimpl}. Let $A > 0$ be any constant and set 
\begin{equation}\label{eq:delta_param_def}
\delta \eqcolon \min\sset{ \frac{1}{4\sqrt{1+t}} \paren{\frac{1}{2\ln 2 \cdot\Rstar} + \eta}^{-1/2}n^{-A - 2\alpha -1}, \frac{1}{\sqrt{2(1+t)}}n^{-(A+2\alpha+1)/2}}\,.
\end{equation}
Notice that since $\forall i, \frac{\lambda_i}{1+2T\lambda_i} \leq \frac 1 {2T}$, we have $\Dstar \leq \frac{n^{2\alpha}}{2T} \leq \frac {n^{2\alpha}}{2\ln 2 \cdot \Rstar}$,
where the last inequality is derived by $\ln 2 \cdot \Rstar = \frac1{2n} \sum_i \ln(1+2\lambda_i T) \leq T$. Plugging in the $\delta$ value in Eq.~\eqref{eq:delta_param_def} into Eq.~\eqref{eq:dist_final_unsimpl}, we obtain 
\begin{align*}
    d_{\Sigma_X}(W,q(\tau) c_i) &\leq n \paren{\Dstar + n^{2\beta} \eta + 2 \delta\sqrt{\Dstar + n^{2\beta} \eta}\cdot n^{\alpha} \sqrt{1+t} + \delta^2n^{2\alpha}(1+t)}\\
    &\leq n \paren{ \Dstar + n^{2\beta} \eta +  \frac 12 n^{-A-1} + \frac12 n^{-A-1}} = n(\Dstar + n^{2\beta} \eta) + n^{-A}\,.
\end{align*}
We denote $c', C_1, C_2, C_3, C_4, C_{12}$ to be explicit constants depending on $\eps, \eta, \Rstar, \alpha$ (but not $W$ or $\Sigma_X$). Now plug in $t = 2^{n\eps} + C^{-1} n \ln 2 (\Rstar + \eps)$. The condition in Eq.~\eqref{eq:dist_final_unsimpl} holds with probability at least 
$$
1- \exp\paren{-2^{n\eps(1-cn^{4(\alpha-\beta)-1})}} - \exp\paren{n\ln 2(\Rstar+\eps)}\exp(-Ct) = 1 - \exp{\paren{-2^{n\eps(1-c'n^{4(\alpha-\beta)-1})}}}\,.
$$
Finally, the rate of this quantization scheme is $R = \Rstar + \eps + \frac 1n \log(1/\delta)$, which we now bound:
\begin{align*}
    \log(1/\delta) &\leq \max\sset{ C_1 \log n + \frac12 \log(1+t) + 2\alpha \log n, C_2 \log n  + \frac12 \log(1+t) + \alpha \log n}\\
    &\leq C_{12} \log n + \frac12 \log(1+t) + 2\alpha\log n\,.
\end{align*}
Plugging in the expression for $t$, we obtain
$$
\log(1+t) \leq C_3 + n \eps  + C_4 \log n\,,
$$
and therefore, 
$$
R \leq \Rstar + \frac 32 \eps + C\cdot \frac{\log n + \alpha\log n}{n}\,.
$$
For sufficiently large $n$, the RHS is $\leq \Rstar + 2\eps$. Moreover, for sufficiently large $n$, our final distortion bound simplifies to $n(\Dstar + n^{2\beta} \cdot 2\eta)$. Setting $\eps = \eps_\circ / 2, \eta = \eta_\circ/2$, we conclude the proof.
\end{proof}

\paragraph{Proof of Claim~\ref{claim:tau_bound}.}
        Since for all $j \in \sqb{n}$, $\Wt_j^2 \leq \|\Wt\|_{\infty}^2$ and $\frac {\lambda_j}{1 + \lambda_j T} \leq \frac 1 {T}$,
        \begin{align*}
            T \sum_{j} \frac{\Wt_j^2 \lambda_j^2}{(1+\lambda_j T)^2} \leq \|\Wt\|_{\infty}^2 \sum_j \frac{\lambda_j}{1+\lambda_j T}\,,
        \end{align*}
        yielding $0 \leq \tau \leq \|\Wt\|_{\infty}$.\qed

\paragraph{Proof of Claim~\ref{claim:tau_quant_bound}.} 
        \begin{align*}
            \sqrt{d_{\Sigma_X}(W, q(\tau)c_i)} &= \sqrt{ \E_{X} \sqb{\b\| (W - q(\tau)c_i)^T X \b\|^2_2} }\\
            &\leq \sqrt{ d_{\Sigma_X}(W, \tau c_i) } + \sqrt{ \E_{X} \sqb{\b\| (\delta_{\tau}\cdot c_i)^T X \b\|^2_2} }\\
            &= \sqrt{ d_{\Sigma_X}(W, \tau c_i) } + \delta_\tau \sqrt{   c_i^T \Sigma_X c_i } \,,
        \end{align*}
        where $\delta_{\tau} = |q(\tau) - \tau| \leq \delta n^\alpha$. By Hanson-Wright inequality (Prop.~\ref{prop:hanson_wright}) for any codeword $j$, the second term can be bounded as
        $$
        \Pr\sqb{ | c_j^T \Sigma_X c_j - \trace \Sigma_X| > nt } \leq 2\exp\paren{-c \min \paren{\frac {t^2n^2} {K^4 \| \Sigma_X\|_F^2}, \frac{tn}{K^2 \|\Sigma_X\|_{op}}}}\,,
        $$
        for universal constants $c, K$, and therefore, with probability at least $1 - \exp(- C tn / \|\Sigma_X\|_{op}) \geq 1 - \exp(- C t)$,
        $$
        \delta_\tau \sqrt{ c_j^T \Sigma_X c_j } \leq \delta n^\alpha \sqrt{n(1+t)}\,.
        $$
        The statement of the Claim follows by a union bound over all $2^{n(\Rstar+\eps)}$ codewords.
        \qed

\subsection{Proof of Lemma~\ref{lem:codebook_success}: Success Probability of Random Gaussianå Code}\label{subsec:proof_codebook_success}

\begin{mycond}{\ref{cond:W_Sigma}}\textnormal{\textbf{(Admissible $W, \Sigma_X$)}}
    We consider $W \in \R^n, \Sigma_X\in\mathbb{S}^n_{+}$ satisfying 
\begin{enumerate}
    \item $\mathrm{SVD}\paren{\Sigma_X} = U \Lambda U^\top$ for $U \in \mathcal{O}_n$ and $\Lambda = \diag\paren{\lambda_1,\dots,\lambda_n} \succeq 0$ with $\trace(\Lambda) = n$;
    \item $\|U^\top W\|_{\infty} \leq n^{\alpha}$ for a known $\alpha > 0$.
\end{enumerate}
\end{mycond}

For $W, C \in \R^n$, $\tau \in \R$, and $\Sigma \in \mathbb{S}^n_{+}$, recall that we define the distortion function to be
\begin{equation}\label{eq:def_e_coebook_lem}
    d_{\Sigma_X}(W, \tau C)  = (W - \tau C)^\top \Sigma_X (W - \tau C)\,. \tag{$E$ def.}
\end{equation}
For admissible $(W, \Sigma_X)$ (Cond.~\ref{cond:W_Sigma}), denote $T = \Trc^\lambda(\Rstar)$ to be a unique solution to $\Rrc^{\lambda}(T) = \Rstar$ and define
        \begin{equation}\label{eq:def_tau_codebook_lem}
        \tau(W, \Sigma_X) = \paren{T \sum_{j} \frac{(U^\top W)_j^2 \lambda_j^2}{(1+\lambda_j T)^2}}^{1/2}\paren{\sum_j \frac{\lambda_j}{1+\lambda_j T}}^{-1/2}\,.\tag{$\tau$ def.}
        \end{equation}

\begin{lemma}[Gaussian Book Success]\label{lem:codebook_success} Let $\eps,\eta >0$ be constants and $\CCC = \sset{c_1,\dots,c_{M= 2^{n(\Rstar + \eps)}} }$, $c_i \sim_{i.i.d.} \N(0, I_n)$ be a randomly generated Gaussian codebook.

For admissible $W, \Sigma_X$ (Cond.~\ref{cond:W_Sigma}) with $\alpha < 1/4$, 
        $$
        \Pr_{\CCC}\sqb{ \frac 1 n \min_{i \in \sqb{M}} d_{\Sigma_X}(W, \tau(W,\Sigma_X) \cdot c_i) \leq  \DDrc^{ \lambda}(U^\top W,\Rstar)+\eta} \geq 1 - \exp\paren{-2^{n\eps(1-c n^{4\alpha - 1})}}\,,
        $$
        where $c = c(\eps,\eta, \Rstar, \alpha)$ is an explicit constant function, $\DDrc^{\lambda}(U^\top W, \Rstar)$ is defined in Eq.~\eqref{eq:rd_d}, and $\tau(W,\Sigma_X)$ is defined in Eq.~\eqref{eq:def_tau_codebook_lem}.
\end{lemma}

The proof of Lemma~\ref{lem:codebook_success} relies on the following bound on a probability that a single codeword $c_i \sim \N(0, I_n)$ achieves small distortion $d_{\Sigma_X}$.

\begin{lemma}\label{lem:codeword_success}
    In the setting of Lemma~\ref{lem:codebook_success}, denote 
    $$
    p_n \coloneq \Pr_{c_i \sim \N(0, I_n)}\sqb{ \frac 1 n d_{\Sigma_X}(W, \tau(W, \Sigma_X)\cdot c_i) \leq \DDrc^{\lambda}(U^\top W, \Rstar) + \eta}\,.
    $$
    Then, 
    $$
    \ln p_n \geq - n \Rstar (1+ O(n^{4\alpha - 1})) = - n \Rstar (1 + o(1))\,.
    $$
\end{lemma}

\begin{proof}(of Lemma~\ref{lem:codebook_success}) Denote $\Dstar = \DDrc^{\lambda}(U^\top W, \Rstar)$, and $\tau = \tau(W, \Sigma_X)$. Given the $\EVD\paren{\Sigma_X} = U\Lambda U^\top$, by definition of $d_{\Sigma_X}$, for all $i \in \sqb{M}$,
        \begin{align*}
            d_{\Sigma_X}(W, \tau \cdot c_i) &= (W - \tau c_i)^\top U\Lambda U^\top (W - \tau c_i)\\
            &= (U^\top W - \tau U^\top c_i)^\top \Lambda (U^\top W - \tau U^\top c_i)\\
            &\eqcolon (\Wt - \tau \ct_i)^\top \Lambda (\Wt - \tau \ct_i)\,,
        \end{align*}
        where we denote $\Wt = U^\top W$ and $\ct_i = U^\top c_i$. Note that $\ct_i \sim \N(0, I_n)$ are independent.
        
        We show that for the optimal choice of $\tau$ in Eq.~\eqref{eq:def_tau_codebook_lem} and $M = 2^{n\paren{\Rstar+\eps}}$, with probability at least $1 - \exp\paren{-2^{n\eps(1-c \|\Wt\|^4_\infty / n)}}$ (where $\|\Wt\|^4_\infty / n \leq n^{4\alpha - 1}$), 
        $$
        \frac 1 n \min_{i\in\sqb{M}} d_{\Sigma_X}(W, \tau c_i) = \frac 1 n \min_{i\in\sqb{M}} \sum_{j} \lambda_j \paren{\Wt_j - \tau \ct_{ij} }^2 \leq \Dstar+\eta\,.
        $$
        By Lemma~\ref{lem:codeword_success}, for any $i\in \sqb{M}$,
        \begin{equation*}
            \ln p_n = \ln \Pr\sqb{ \frac 1 n \sum_{j} \lambda_j \paren{\Wt_j - \tau \ct_{ij} }^2 \leq \Dstar + \eta} \geq -n\Rstar(1+O(n^{4\alpha - 1}))\,.
        \end{equation*}
        Then, since $c_1,\dots,c_M$ are independent, 
        \begin{align*}
            \Pr\sqb{\frac 1 n \min_{i \in M} d_{\Sigma_X}(W, \tau c_i) \leq \Dstar+\eta } \geq 1 - \paren{1 - p_n}^M &\geq 1 - \exp\paren{ - 2^{-n\Rstar(1+O\paren{\|\Wt\|^4/n})} 2^{n\Rstar + n \eps}}\\
            &= 1 - \exp\paren{-2^{n\eps(1-O\paren{\|\Wt\|^4/n})}}\\
            &\geq 1 - \exp\paren{-2^{n\eps(1-O\paren{n^{4\alpha-1}})}}\,,
        \end{align*}
        which concludes the proof of Lemma~\ref{lem:codebook_success}. 
    \end{proof}

    \begin{proof}(of Lemma~\ref{lem:codeword_success}) Denote $\Rnat = \Rstar \ln 2$, $\Dstar = \DDrc^{\lambda}(U^\top W, \Rstar)$, and $\tau = \tau(W, \Sigma_X)$.
        First, for all $j \in \sqb{n}$ denote $\mu_j = \mathcal{L}\paren{ \lambda_j \paren{\Wt_j - \tau \ct_{ij} }^2 }$, where $\ct_{ij}\sim\N(0,1)$ and rewrite 
        $$
        p_n = \Pr_{X_j \sim \mu_j}\sqb{\frac 1 n \sum_j X_j \leq \Dstar + \eta}\,.
        $$
        Let $t = - T/(2\tau^2)$, where $T, \tau$ are defined in Eq.~\eqref{eq:def_tau_codebook_lem}. We have  
        \begin{align*}
            \ln \E_{X\sim\mu_j}\sqb{e^{tX}} = \frac{t \Wt_j^2 \lambda_j}{1 - 2 t \lambda_j \tau^2} - \frac12 \ln \paren{1 - 2 t \lambda_j \tau^2} \eqcolon \Phi_j(t)\,.
        \end{align*}
        A direct computation yields 
        \begin{align}
            \Phi'_j(t) &= \frac{\Wt_j^2 \lambda_j}{(1-2t \lambda_j \tau^2)^2} + \frac{\lambda_j \tau^2}{1 - 2t\lambda_j \tau^2}\label{eq:der_t}\,.
        \end{align}
        For each $j \in \sqb{m}$, define a new probability measure $\tilde \mu_j$ as
        $
        \frac{d\tilde\mu_j}{d\mu_j}(x) = e^{t x - \Phi_j(t)}\,.
        $
        We have $\int_{\R} d\tilde \mu_j = e^{-\Phi_j (t)} \int_{\R} e^{tx}d\mu_j = 1$. Since the exponential tilting above is applied to a noncentral $\chi^2$ distribution $\mu$, the transformation simply shifts/rescales the parameters of the $\chi^2$ distribution. We can explicitly compute
        $$
        \tilde \mu_j = \mathcal{L}\paren{ \N\paren{m_j, s_j^2}^2 }\,,\qquad\text{where } m_j = \frac{\Wt_j \lambda_j^{1/2}}{1 - 2t \lambda_j \tau^2} \text{ and } s_j^2 = \frac{\lambda_j\tau^2}{1-2t\lambda_j\tau^2}\,.
        $$
        Plugging in the choice of $\tau$ in Eq.~\eqref{eq:def_tau_codebook_lem} and $T = -2 t \tau^2$ into Eq.~\eqref{eq:der_t}, we obtain 
        \begin{align*}
        \E_{X \sim \tilde\mu}\sum_j X_j &= \sum_j \int x e^{tx - \Phi_j(t)} d \mu_j(x)  = \sum_j \frac d {dt} \ln \int e^{tx} d \mu_j(x) = \sum_j \Phi_j'(t) \\
        &\qquad\qquad= \sum_j \frac{\Wt_j^2 \lambda_j}{1 - 2t \lambda_j \tau^2} = \sum_j \frac{\Wt_j^2 \lambda_j}{1 + \lambda_j T} = n \Dstar\,.
        \end{align*}
        Then, 
        \begin{align*}
            D(\tilde\mu \| \mu) = \sum_j D(\tilde\mu_j \| \mu_j) = \sum_j t \Phi_j'(t) - \Phi_j = \frac12 \sum_j \ln (1- 2t\lambda_j \tau^2) = \frac12 \sum_j \ln (1 + \lambda_j T) = n\Rnat\,.
        \end{align*}
        Recall that we defined $p_n = \Pr_{X_j \sim \mu_j}\sqb{\sum_j X_j \leq n \paren{\Dstar + \eta}}$ and let $\tilde p_n = \Pr_{X_j \sim \tilde\mu_j}\sqb{\sum_j X_j \leq n \paren{\Dstar + \eta}}$. By DPI,
        \begin{align*}
            n\Rnat = D(\tilde\mu \| \mu) \geq d\paren{ \tilde p_n \| p_n} \geq - h\paren{\tilde p_n} + \tilde p_n \ln \frac 1 {p_n} \geq - \ln 2 + \tilde p_n \ln \frac 1 {p_n}\,,
        \end{align*}
        which yields 
        \begin{equation}\label{eq:p_n_lower_bd}
        \ln p_n \geq \frac{-n \Rnat - \ln 2}{\tilde p_n}\,.
        \end{equation}
        By Chebyshev's inequality, 
        $$1 - \tilde p_n = 1- \Pr_{X_j \sim \tilde\mu_j}\sqb{ \sum_j X_j \leq \E \sum_j X_j + n\eta} \leq \frac{\Var\sqb{\sum_j X_j}}{n^2\eta^2}\,.$$
        In what follows we show that 
        \begin{equation}\label{eq:var_bound}
            \Var\sqb{\sum_j X_j} = O\paren{n \|\Wt\|_\infty^4} = o(n^2)\,,\tag{Var}
        \end{equation}
        yielding $1 - \tilde p_n = O\paren{\|\Wt\|_\infty^4/n} = o(1)$. Plugging this into Eq.~\eqref{eq:p_n_lower_bd}, we get 
        $$
        \ln p_n \geq \frac{-n\Rnat - \ln 2}{1 - O\paren{\|\Wt\|_\infty^4/n}} =  -n\Rnat(1+O\paren{\|\Wt\|_\infty^4/n})
        $$
        for sufficiently large $n$. 

        It remains to obtain the bound in Eq.~\eqref{eq:var_bound}. From the definition of $\Rnat$ and $\frac{x}{1+x}\leq \ln(1+x) \leq x$ for $x > -1$, 
        $$
        |t|\tau^2 \sum_j \frac{\lambda_j}{1-2t\lambda_j\tau^2} \leq 
        n \Rnat = \frac12 \sum_j \ln(1- 2 t \lambda_j\tau^2) \leq |t|\tau^2 n\,. 
        $$
        Using $\frac {\lambda_j} {1-2t\lambda_j\tau^2} \leq \frac 1 {2|t| \tau^2}$ and $(1-2t\lambda_j \tau^2)^2\geq 4|t|\lambda_j \tau^2$, we obtain
        \begin{align*}
            \Var\sqb{X_j} = \Lambda_j''(t) &= \frac{4\Wt_j^2 \lambda_j^2 \tau^2}{(1-2t\lambda_j \tau^2)^3} + \frac{2 \lambda_j^2 \tau^4}{(1-2 t \lambda_j \tau^2)^2}\\
            &\leq \frac{4\Wt_j^2 \lambda_j^2 \tau^2}{(1-2t\lambda_j \tau^2) \cdot 4 |t| \lambda_j \tau^2} + \frac {\tau^2} {|t|} \cdot \frac{ \lambda_j}{1-2 t \lambda_j \tau^2} \\
            &= \frac 1 {|t|}\cdot  \frac{\Wt_j^2 \lambda_j}{1-2t\lambda_j \tau^2}  + \frac {\tau^2} {|t|} \cdot \frac{ \lambda_j}{1-2 t \lambda_j \tau^2}\,,
        \end{align*}
        and therefore, 
        \begin{equation}\label{eq:bd_var}
        \Var\sqb{\sum_j X_j} \leq n\cdot \paren{\frac{\Dstar} {|t|} + \frac{\Rnat}{|t|^2}  }\,.
        \end{equation}
        We show a lower bound on parameter $|t|$. From the way we choose $\tau$ in Eq.~\eqref{eq:tau},
        \begin{align}\label{eq:t_lower_bd}
            |t| = \frac{\sum_j \frac {\lambda_j} {1-2t\lambda_j\tau^2} }{\sum_j \frac {2\Wt_j^2 \lambda_j^2} {(1-2t\lambda_j\tau^2)^2}}\geq \frac{2 \Rnat}{\|\Wt\|_\infty^2} \frac{\sum_j \frac {\lambda_j \|\Wt\|_\infty^2} {1-2t\lambda_j\tau^2} }{\sum_j \frac {2\Wt_j^2 \lambda_j} {1-2t\lambda_j\tau^2}} \geq \frac{\Rnat}{\|\Wt\|_\infty^2}\,,
        \end{align}
        where we used $\frac {\lambda_j} {1-2t\lambda_j\tau^2} \leq \frac 1 {2|t| \tau^2} \leq \frac 1 {2 \Rnat}$. Plugging \eqref{eq:t_lower_bd} into the bound for variance in \eqref{eq:bd_var} and using $\Dstar\leq \frac {\|\Wt\|^2_{\infty}} {T} \leq \frac {\|\Wt\|^2_{\infty}} {2\Rnat}$, we obtain 
        $$
        \Var\sqb{\sum_j X_j} \leq n\cdot \|\Wt\|_\infty^4\paren{\frac{1 }{2(\Rnat)^2} + \frac{1}{\Rnat}  } = O\paren{n \|\Wt\|_\infty^4 }\,,
        $$
        which is the desired inequality in Eq.~\eqref{eq:var_bound}. Here we used that $\Rnat = \Theta(1)$.
        \qed

\end{proof}

\section{Random Coding: Gaussian Isotropic $W$}\label{sec:rc_gaussian_W}

In this section we describe quantization of a Gaussian isotropic vector $W \sim \N(0, I_n) \in\R^n$ using random coding.

\paragraph{Rate-Distortion Function.}  Let $\lambda = (\lambda_1,\dots,\lambda_n)^\top \in \R^n$ be such that $\lambda_i\geq 0$ for all $i\in\sqb{n}$ and $\sum_i \lambda_i = n$. We define \emph{random-coding rate-distortion function} in dimension $n$ to be the following parametric curve for $T > 0$:
\begin{equation}\label{eq:rd_rc_G}
    \Drc^{\lambda}(T) = \frac 1n \sum_{i=1}^n \frac{\lambda_i}{1+\lambda_i T}\qquad\text{and}\qquad \Rrc^{\lambda}(T) = \frac 1{2 n}\sum_{i=1}^n \log(1+\lambda_i T)\,.\tag{RDRC}
\end{equation}
Throughout, $\log$ denotes $\log_2$ and $\ln$ denotes the natural logarithm.
Denote $\Trc^\lambda(R)$ to be a unique value $T$, s.t. $\Rrc^\lambda(T) = R$. Let
\begin{equation}\label{eq:rd_d_G}
    \DDrc^{\lambda}(R) \triangleq \Drc^{\lambda}(\Trc^\lambda(R))\,.\tag{$\DDrc$}
\end{equation}
We consider a task of quantizing a vector $W \sim\N(0, I_n) \in \R^n$ under a distortion function 
\begin{equation}\label{eq:dist_fn_app_G}
d_{\Sigma_X}(W, \Wh) = \E_{X} \sqb{\b\| \paren{W - \hat W}^T X \b\|^2_2} = (W- \Wh)^\top \Sigma_X (W - \Wh)\,,\tag{$d_{\Sigma_X}$}
\end{equation}
where $X \in \R^n$ is a random vector with $\E XX^\top = \Sigma_X \in \mathbb{S}_+^n$. Here we assume $\EVD\paren{\Sigma_X} = U \Lambda U^\top$ for $U \in \mathcal{O}_n$ and $\Lambda = \diag\paren{\lambda_1,\dots,\lambda_n}\succeq 0$ with $\trace(\Lambda) = n$. Our goal is to obtain an upper bound on $\E_{W} \sqb{d_{\Sigma_X}(W, \Wh)}$.

In Theorem~\ref{thm:fixed_sigma_G} we show that there exists an encoder-decoder that, for any constant target rate $\Rstar > 0$ and admissible $\Sigma_X$ as above, achieves rate $\Rstar$ and expected distortion $\DDrc^{\lambda}(\Rstar)$ asymptotically, with high probability over the randomness $S$ shared between encoder and decoder (from which the codebook is generated).

\begin{theorem}[Achievability of RC RD for Gaussian Input: Nonasymptotic Guarantee]\label{thm:fixed_sigma_G}

    Fix any constants $\Rstar, \eps_\circ, \eta_\circ, B > 0$. There exists an encoder $f: \R^n \times \mathbb{S}_+^n \times [0,1] \to \sqb{2^{nR}}$, a decoder $g: \sqb{2^{nR}} \times [0,1] \to \R^n$ with $R \leq \Rstar + \eps_\circ$, and a (shared) random variable $S\in \sqb{0,1}$ with the following property.
    
    For $W\sim\N(0,I_n)$, any $\Sigma_X$ (with $\EVD(\Sigma_X) = U \Lambda U^\top$ and $\Lambda = \diag(\lambda_1,\dots,\lambda_n)$) as above we set $\Wh = g(f(W, \Sigma_X, S), S)$. Then for any sufficiently large $n \geq n_0 = n_0 (\eps_\circ,\eta_\circ, \Rstar, B)$ we have
    $$
    \Pr_S\sqb{\E_W \sqb{\frac 1 n d_{\Sigma_X}(W,\Wh)} \leq \DDrc^{\lambda}(\Rstar) + \eta_\circ} \geq 1 - \exp\paren{- n^{B}}\,,
    $$
    where $d_{\Sigma_X}(W,\Wh)$ is the distortion function in Eq.~\eqref{eq:dist_fn_app_G}.
\end{theorem}

In Theorem~\ref{thm:union} we show that, in fact, the quantization scheme in Theorem~\ref{thm:fixed_sigma_G} achieves the rate-distortion guarantee simultaneously for all $\Sigma_X \in \mathbb{S}_+^n$ with $\trace(\Sigma_X) = n$, with high probability over the shared randomness (i.e., the codebook $\CCC$)

\begin{mythm}{\ref{thm:union}}\textnormal{\textbf{(Achievability of RC RD for Gaussian Input: Worst-Case $\Sigma_X$)}}
    Fix any constants $\Rstar, \eps_\circ, \eta_\circ, B > 0$. There exist an encoder $f: \R^n \times \mathbb{S}_{+}^n \times[0,1] \to \sqb{2^{nR}}$, a decoder $g: \sqb{2^{nR}} \times [0,1] \to \R^n$ with $R \leq \Rstar + \eps_\circ$, and a (shared) random variable $S\in [0,1]$ with the following property. 
    
    For $W\sim\N(0,I_n)$ and $\Sigma_X\in \mathbb{S}^n_+$ we set $\Wh(\Sigma_X) = g\paren{f(W, \Sigma_X, S), S}$. Then, for sufficiently large $n \geq n_0 = n_0(\eps_\circ,\eta_\circ, \Rstar, B)$, we have 
    $$
    \Pr_S\sqb{\sup_{\substack{\Sigma_X \in \mathbb{S}_+^n\\ \trace(\Sigma_X) = n}} \paren{\frac 1 n \E_W \sqb{d_{\Sigma_X}(W,\Wh(\Sigma_X))} -  \DDrc^{\spec(\Sigma_X)}(\Rstar)} \leq \eta_\circ} \geq 1 - \exp\paren{- n^B}\,,
    $$
    where $d_{\Sigma_X}(W,\Wh)$ is the distortion function in Eq.~\eqref{eq:dist_fn_app_G}.
\end{mythm}

\subsection{Proof of Theorem~\ref{thm:fixed_sigma_G}}

The proof uses the result of Theorem~\ref{thm:fixed_sigma}, which obtains a distortion guarantee for any fixed vector $W\in \R^n$. We adjust the quantization scheme to not rely on the Theorem~\ref{thm:fixed_sigma} assumption $\|U^\top W\|_\infty \leq n^\alpha$ (Cond.~\ref{cond:W_Sigma}): in the unlikely case of large $\|U^\top W\|_\infty$, the decoder returns $0$. 

\paragraph{Quantization Scheme.} 
    Fix any constant $\alpha \in (0,1/4)$.
    From the shared randomness $S$, $f,g$ generate a Gaussian codebook $\CCC = \sset{c_1,\dots,c_{M=2^{n(\Rstar+\eps)}}}$ for $\eps = \eps(\eps_\circ)$ to be chosen later. 
    Denote $\Wt = U^\top W$. We define:
    \begin{itemize}
        \item \textbf{Encoder $f$:} Let $T \coloneq \Trc^\lambda(\Rstar)$ and define the scaling parameter $\tau = \tau(\Wt, \lambda)$ as
        \begin{equation}\label{eq:tau_G}
        \tau = \begin{cases}
            \paren{T \sum_{j} \frac{\Wt_j^2 \lambda_j^2}{(1+\lambda_j T)^2}}^{1/2}\paren{\sum_j \frac{\lambda_j}{1+\lambda_j T}}^{-1/2} &\text{if }\|\Wt\|_\infty \leq n^{\alpha}\\
            0 &\text{otherwise.}
            \end{cases}
            \tag{$\tau$ def.}
        \end{equation}
        We set $f$ to be a tuple
        $$
        f(W, \Sigma_X, S) = \paren{\argmin_{i \in \sqb{M}} d_{\Sigma_X}(W, \tau c_i), q(\tau)}\,,
        $$
        where $q(\tau) = \delta \|\Wt\|_{\infty} \lfloor \tau/(\delta \|\Wt\|_{\infty}) \rceil$ is a rounding quantization scheme of precision $\delta$ and recall the distortion function 
        $$
        d_{\Sigma_X}(W, C) = (W - C)^\top \Sigma_X (W - C)\,.
        $$
        \item \textbf{Decoder $g$:}
        $$
        g(i, q(\tau), S) = q(\tau) \cdot c_i\,.
        $$
    \end{itemize}

\paragraph{Rate-Distortion Bound.}  
Denote 
    In the quantization scheme above, 
    $$
    R = \underbrace{\Rstar + \eps}_{\text{Gauss. codebook}} + \underbrace{\frac1n \log(1/\delta)}_{\tau\text{ quant.}}\,.
    $$
    The rest of the proof is to obtain a high probability bound on the resulting distortion that, given $\Wh = g( f(W,\Sigma_X, S),S) = q(\tau)\cdot c_i$, can be expressed as:
    \begin{align*}
        d_{\Sigma_X}(W,\Wh) = d_{\Sigma_X}(W,q(\tau) c_i) =(W - q(\tau) c_i)^\top \Sigma_X (W - q(\tau) c_i)\,.\label{eq:dist_expand2}
    \end{align*}
    In what follows, denote $\Dstar = \DDrc^{\lambda}(\Rstar)$. Let $\eta = \eta(\eta_\circ)$ be a constant to be chosen later. By Theorem~\ref{thm:fixed_sigma}, with an appropriate choice of parameter $\delta = \delta (\eps, \eta, \Rstar, \alpha, n)$ in the quantization scheme above, we have for any fixed $W$ such that $\|U^\top W\|_\infty \leq n^\alpha$ and sufficiently large $n \geq n_0 = n_0(\eps,\eta,\alpha,\Rstar)$,
    $$
    \Pr_{\CCC}\sqb{\frac 1 n d_{\Sigma_X}(W, \Wh) \leq \DDrc^{\lambda}(U^\top W, \Rstar) + \eta} \geq 1 - \exp\paren{-2^{n\eps(1-cn^{4\alpha - 1})}}\,,
    $$
    where $c = c(\eps,\eta,\alpha,\Rstar)$. 
    At the same time, the rate can be bounded as $R \leq \Rstar + \eps_\circ$ (for an appropriately chosen $\eps = \eps(\eps_\circ)$). Denote the event of the codebook failure as $F(W, \CCC) : \R^n \times \paren{\R^n}^{M}\to \sset{0,1}$:
    \begin{equation}\label{eq:codebook_failure}
        F(W,\CCC) = \sset{ \|U^\top W\|_\infty \leq n^{\alpha}\text{ and } \frac 1 n d_{\Sigma_X}(W, \Wh) > \DDrc^{\lambda}(U^\top W, \Rstar) + \eta }\,.\tag{$F$ def.}
    \end{equation}
    Rewriting the result of Theorem~\ref{thm:fixed_sigma}, we obtain that 
    $$\forall W\in \R^n:\ \|U^\top W\|_\infty \leq n^{\alpha}: \qquad \Pr_{\CCC}\sqb{F(W,\CCC)} \leq \exp\paren{-2^{n\eps(1-cn^{4\alpha - 1})}}\,.$$
    Denote $E_\CCC$ to be the event $E_\CCC = \sset{\forall i\in\sqb{M}:\ \| c_i\|_2 \leq n^{B}}$ for some fixed constant $B>10$ and note that $$
    \Pr_\CCC\sqb{ E_\CCC^c} \leq M \cdot e^{-C n^{2B}} = e^{n\ln2 (\Rstar + \eps) - C n^{2B}} \leq  \exp\paren{- c' n^{2B}}\,,
    $$
    for constant $c' = c'(\eps, \Rstar)$ and sufficiently large (constant) $n$. We denote $\CCC | E_\CCC$ to be the distribution of $\CCC$ conditioned on $E_\CCC$ and note that, since $\Pr_\CCC\sqb{E_\CCC} \geq 1/2$, Theorem~\ref{thm:fixed_sigma} in fact yields
    $$\forall W\in \R^n:\ \|U^\top W\|_\infty \leq n^{\alpha}: \qquad \Pr_{\CCC|E_\CCC}\sqb{F(W,\CCC)} \leq 2\exp\paren{-2^{n\eps(1-cn^{4\alpha - 1})}}\,.$$
    We show in steps 1-2 that for any constant $A > 0$ and sufficiently large $n$,
    $$
    \Pr_{\CCC|E_\CCC}\sqb{\E_W \sqb{d_{\Sigma_X}(W,\Wh)} \leq n (\DDrc^{\lambda}(\Rstar) + \eta) + n^{-A}} \geq 1 - \exp\paren{-2^{n\eps(1-c'n^{4\alpha - 1})}}\,,
    $$
    and then conclude the proof in step 3 via a union bound.

    \paragraph{Step 1: Bound on $\Pr_{\CCC | E_\CCC}\sqb{\Pr_{W\sim\N(0,I_n)}\sqb{F(W,\CCC) | \CCC} > \pstar }$.} Denote the event $E_W = \sset{\|U^\top W\|_{\infty} \leq n^{\alpha}}$. All expressions below assume $W \sim \N(0, I_n)$.
    \begin{align*}
        \E_{\CCC | E_\CCC}\sqb{\Pr_{W}\sqb{F(W,\CCC) | \CCC} } = \Pr_{W, \CCC | E_\CCC} \sqb{F(W,\CCC)}= \Pr_{W}\sqb{\Pr_{\CCC | E_\CCC} \sqb{F(W,\CCC)|W}}\leq \Pr_{W} \sqb{E_W} \cdot 2\exp\paren{-2^{n\eps(1-cn^{4\alpha - 1})}}\,.
    \end{align*}
    Then, by Markov's inequality, 
    \begin{equation}\label{eq:markovs}
        \Pr_{\CCC|E_\CCC}\sqb{\Pr_{W}\sqb{F(W,\CCC) | \CCC} > \pstar } \leq \frac{\Pr_{W} \sqb{E_W} \cdot 2\exp\paren{-2^{n\eps(1-cn^{4\alpha -1 })}}}{\pstar}\,.
    \end{equation}

    \paragraph{Step 2: Bound on $\E_W\sqb{d_{\Sigma_X}(W,\Wh)}$ in case of codebook success in Eq.~\eqref{eq:markovs}.} Fix any codebook $\CCC$ for which $\Pr_{W}\sqb{F(W,\CCC) | \CCC} \leq \pstar$ and $E_\CCC$ hold. Expand 
    \begin{align*}
        \E_{W}\sqb{d_{\Sigma_X}(W,\Wh)} = &\underbrace{\E_{W}\sqb{d_{\Sigma_X}(W,\Wh) \I_{E_W}}}_{\one} + \underbrace{\E_{W}\sqb{d_{\Sigma_X}(W,\Wh) \I_{E_W^c}}}_{\two}\,.
    \end{align*}
    First, 
    \begin{align*}
    E_{W}\sqb{d_{\Sigma_X}(W,\Wh) \I_{E_W}} \leq  &\pstar \E_{W}\sqb{d_{\Sigma_X}(W,\Wh)\I_{E_W} | F(W,\CCC)} + \\
    &\qquad\qquad\Pr_{W}\sqb{F(W,\CCC)^c}\E_{W}\sqb{d_{\Sigma_X}(W,\Wh) \I_{E_W} | F(W,\CCC)^c}\,.
    \end{align*}
    For any $W$, we have, since $q(\tau) \leq n^{\alpha}$ by Claim~\ref{claim:tau_bound},
    $$
    d_{\Sigma_X}(W,\Wh) = (W-q(\tau)c_i)^\top \Sigma_X (W - q(\tau)c_i) \leq (\|U^\top W\|_\infty+ n^{\alpha}\|U^\top c_i\|_\infty)^2 n\,,
    $$
    and therefore, since the above is $\leq n^{3B}$ for $W \in E_W$ and $\CCC \in E_\CCC$,
    \begin{equation}\label{eq:I_1}
    \pstar E_{W}\sqb{d_{\Sigma_X}(W,\Wh)\I_{E_W} | F(W,\CCC)} \leq \pstar \cdot n^{3B}\,.\tag{$\one.1$}
    \end{equation}
    Moreover, by the definition of $F(W,\CCC)$ in Eq.~\eqref{eq:codebook_failure}, 
    \begin{align}
        \Pr_{W}\sqb{F(W,\CCC)^c} \E_{W}\sqb{d_{\Sigma_X}(W,\Wh)\I_{E_W} | F(W,\CCC)^c} &\leq n\E_W\sqb{\DDrc^{\lambda}(U^\top W, \Rstar)\cdot \I_{F(W,\CCC)^c} \I_{E_W}} + n\eta\notag\\
        &\leq \E_W\sqb{\sum_i \frac{(U^\top W)^2_i \lambda_i}{1+T\lambda_i} \I_{E_W}} + n\eta \notag\\
        &\leq \E_W\sqb{\sum_i \frac{(U^\top W)^2_i \lambda_i}{1+T\lambda_i}} + n\eta \notag\\
        &= n (\Dstar + \eta)\,.\tag{$\one.2$}\label{eq:I.2}
    \end{align}
    In the above we used that $\Rrc^{\lambda}(T)$ is independent of vector $U^\top W$, and therefore both $\DDrc^{ \lambda}(U^\top W, \Rstar)$ and $\DDrc^{\lambda}(\Rstar)$ share the same parameter $T>0$.
    Combining \eqref{eq:I_1} and \eqref{eq:I.2}, we obtain 
    \begin{equation}\label{eq:I}
    \one \leq  n (\Dstar + \eta) + \pstar \cdot n^{3B}\,.\tag{$\one$}
    \end{equation}
    Now in the case of $E_W^c$, we have $\tau = 0$, and therefore, $d_{\Sigma_X}(W,\Wh) = W^\top \Sigma_X W \leq n\|U^\top W\|_\infty^2 $. Note that $U^\top W \sim\N(0,I_n)$, so by the standard Gaussian tail bound, 
    \begin{align}\label{eq:II}
        \two \leq n \E_W\sqb{\|U^\top W\|_\infty^2 \cdot \I_{ \|U^\top W\|_\infty > n^{\alpha}}} =n \int_{t = n^{\alpha}}^{\infty} t^2 \cdot 2n e^{-t^2/2} \leq 4 n^{2+\alpha} e^{-n^{2\alpha}/2}\,.\tag{$\two$}
    \end{align}
    Combining \eqref{eq:I} and \eqref{eq:II} together, we obtain 
    $$
    \E_W \sqb{d_{\Sigma_X}(W,\Wh)} \leq n (\Dstar + \eta) + \pstar \cdot n^{3B} + 4 n^{2+\alpha} e^{-n^{2\alpha}/2} \leq n (\Dstar + \eta) + n^{-A}\,,
    $$
    for sufficiently large constant $n$ and the choice $\pstar = e^{-n}$. 

    \paragraph{Step 3: Bound on the probability of successful codebook $\CCC$.} Recall that 
    $
    \Pr_\CCC\sqb{ E_\CCC^c} \leq  \exp\paren{- c' n^{2B}}
    $.
    From Eq.~\eqref{eq:markovs} and the choice of $\pstar$ above, we have for sufficiently large constant $n$, 
    $$
    \Pr_{\CCC|E_\CCC}\sqb{\Pr_{W}\sqb{F(W,\CCC) | \CCC} > \pstar } \leq \exp\paren{-2^{n\eps(1-cn^{4\alpha - 1})}+n} \leq  \exp\paren{- c' n^{2B}}\,.
    $$
    Then, by a union bound, the probability to select a good codebook is at least $1 - \exp\paren{- c'' n^{2B}}$, so we conclude (since for sufficiently large $n$, $n^{-A} < n\eta$), 
    $$
    \Pr_\CCC\sqb{\E_W \sqb{\frac 1 n d_{\Sigma_X}(W,\Wh)} \leq \Dstar + 2\eta} \geq 1 - \exp\paren{- c'' n^{2B}}\,,
    $$
    from which the statement of the theorem follows.

\subsection{Proof of Theorem~\ref{thm:union}}\label{subsec:proof_of_union}

We show Theorem~\ref{thm:union} by applying the results of Theorem~\ref{thm:fixed_sigma_G} for fixed $\Sigma_X$ with a covering argument. The quantization scheme is the same as in the proof of Theorem~\ref{thm:fixed_sigma_G}; the technical challenge of Theorem~\ref{thm:union} is to show that is succeeds simultaneously with high probability for all $\Sigma_X\in \mathbb{S}_+^n$ with $\trace(\Sigma_X) = n$. 

We first recall some notation. In the quantization scheme of Theorem~\ref{thm:fixed_sigma_G} we have a Gaussian codebook $\CCC = \sset{c_1,\dots,c_M}$ for $M = 2^{n (\Rstar + \eps)}$ (where $\eps = \eps(\eps_\circ)$ is a constant) that is generated from the shared randomness $S\in[0,1]$.
We denote $\EVD(\Sigma_X) = U \Lambda U^\top$ for $\Lambda = \diag(\lambda_1,\dots,\lambda_n)$, 
and recall the distortion
$$
d_{\Sigma_X}(W, C) = (W - C)^\top \Sigma_X (W- C)\,.
$$
The quantizer then sets the optimal scaling factor $\tau = \tau(U^\top W, \Lambda)$ in Eq.~\eqref{eq:tau_G} and returns a tuple $$f(W, \Sigma_X, S) = \paren{\argmin_{i \in \sqb{M}} d_{\Sigma_X}(W, \tau c_i), q(\tau)}\,,$$
where $q(\tau)$ is a simple rounding quantizer. The decoder is defined to be 
$$
g((i, q(\tau)), S) = q(\tau)\cdot c_i\,.
$$
Finally, for some constant $B>10$, we denote $$E_\CCC = \sset{\forall i\in\sqb{M}:\ \| c_i\|_2 \leq n^{B}}\,,\qquad \Pr_\CCC\sqb{ E_\CCC^c} \leq  \exp\paren{- c' n^{2B}} $$
and, denoting the conditional distribution $\CCC | E_\CCC$, the proof of Theorem~\ref{thm:fixed_sigma_G} shows for any $\alpha \in (0, 1/4)$ and $\eta = \eta(\eta_\circ)$ to be chosen later,
    \begin{equation}\label{eq:thm_fixed_cond_result}
    \Pr_{\CCC|E_\CCC}\sqb{\E_W \sqb{\frac1n d_{\Sigma_X}(W,\Wh)} \leq \DDrc^{\lambda}(\Rstar) + \eta} \geq 1 - \exp\paren{-2^{n\eps(1-c'n^{4\alpha - 1})}}\,,
    \end{equation}
for sufficiently large $n \geq n_0 = n_0(\eps,\eta,\Rstar,\alpha)$. Simultaneously, it is shown that the scheme above achieves final rate $R \leq R + \eps_\circ$ for the appropriately chosen $\eps = \eps(\eps_\circ)$.

\paragraph{Step 1: Covering for $U \in \mathcal{O}_n$ and $\Lambda$.}

\begin{fact}[Covering of $\mathcal{O}_n$]\label{fact:covering} For a universal constant $c$,
    $$N(\gamma, \mathcal{O}_n, \|\cdot\|_{op}) \leq (c/\gamma)^{n^2}.$$
    Consequently, there exists a $\gamma$-covering $U_1, \dots, U_N \in \mathcal{O}_n$ such that for all $U \in \mathcal{O}_n$, 
    $$
    \min_{i \in N} \| U_i - U\|_{op} \leq \gamma\,,
    $$
    and $N \leq (c/\gamma)^{n^2} $.
\end{fact}
\begin{fact}[Covering of $\Lambda$]\label{fact:covering_Lambda} For a universal constant $c$,
    $$
    N\paren{\delta/\sqrt{d}; B_1^d, \|\cdot\|_2} \leq \paren{c+ c/\delta}^d\,. 
    $$
    Consequently, there exists a
    $(\tilde \gamma \sqrt{n})$-covering $\tilde s_1, \dots, \tilde s_{N'} \in \sset{s \in \R_{+}^{n}:\, \sum_j s_j \leq n}$ such that for all $\lambda \in \R^n_+$ with $\sum_j\lambda_j \leq n$,
    $$\min_{i\in N'} \|\tilde s_i - \lambda\|_2 \leq \tilde \gamma\sqrt{n}\,,$$
    and $N' \leq (c/\tilde\gamma)^n$ for a universal constant $c$.
\end{fact}
Setting $s_i \coloneq \argmin_{s \in \R_+^n: \sum_j s_j = n} \|s - \tilde s_i\|_2$ and applying Fact~\ref{fact:covering_Lambda} with $\tilde \gamma = \gamma/2$, we obtain:
\begin{corollary}[Of Fact~\ref{fact:covering_Lambda}]
    There exists a
    $(\gamma \sqrt{n})$-covering $s_1, \dots, s_{N'} \in \sset{s \in \R_{+}^{n}:\, \sum_j s_j = n}$ such that for all $\lambda \in \R^n_+$ with $\sum_j\lambda_j = n$,
    $$\min_{i\in N'} \|s_i - \lambda\|_2 = \gamma\sqrt{n}\,,$$
    and $N' \leq (c/\gamma)^n$ for a universal constant $c$.
\end{corollary}

Consider the $\gamma$-covering $U_1, \dots, U_N \in \mathcal{O}_n$ and $(\gamma_2 \sqrt{n})$-covering $s_1, \dots, s_{N'} \in \sset{s \in \R_{+}^{n}:\, \sum_j s_j = n}$ from the facts above for $\gamma$ to be chosen later. Denote $\Lambda_i \coloneq \diag(s_i)$. By a union bound applied to Eq.~\eqref{eq:thm_fixed_cond_result},
$$
\Pr_{\CCC|E_\CCC}\sqb{\sup_{\Sigma_X = U_i \Lambda_j U_i^\top} \paren{\E_W \sqb{\frac 1 n d_{\Sigma_X}(W,\Wh)} - \DDrc^{s_j}(\Rstar)} \leq \eta } \geq 1 - NN'\exp\paren{-2^{n\eps(1-c'n^{4\alpha - 1})}}\,.
$$

\paragraph{Step 2: Perturbation bound on $\E_W \sqb{d_{\Sigma_X}(W,\Wh)}$.} Fix any constant $A > 0$ (say $A = 10$). We show the following distortion perturbation results.
\begin{claim}[Distortion Perturbation]\label{claim:dist_perturb_union}
    Let $\Lambda = \diag(\lambda), \Lambdat = \diag(\lambdat)$ 
    be s.t. $\Lambda, \Lambdat \succeq 0$ and $\trace(\Lambda)=\trace(\Lambdat) = n$ and $U, \Ut \in \mathcal{O}_n$ satisfy
    $$
    \|\Lambda - \Lambdat\|_2 \leq \gamma\sqrt{n} \qquad\text{and}\qquad \|U - \Ut\|_{op} \leq \gamma\,.
    $$
    Denote $$ \Sigma_X = U \Lambda U^\top \quad\text{and}\quad \Sigmat_X = \Ut \Lambdat \Ut^\top\,.$$ If $\CCC = \sset{c_1,\dots,c_M} \in \paren{\R^n}^M$ satisfies $\max_{i\in\sqb{M}} \|c_i\|_{2} \leq n^{B}$, then, in the setup above,  
    $$
     \E_{W\sim \N(0, I_n)} \sqb{d_{\Sigmat_X}(W,\Wh)} \leq \E_{W\sim \N(0, I_n)} \sqb{d_{\Sigma_X}(W,\Wh)} + \frac12 n^{-A} + \gamma n^{O(B)} \cdot 2^{4n\Rstar}\,,
    $$
     for any sufficiently large (constant) $n$, where $O(B)$ hides constants in $A, B, \Rstar$. 
\end{claim}
Plugging in $\gamma = 2^{-5n\Rstar}$, the bound in Claim~\ref{claim:dist_perturb_union} simplifies to $n^{-A}$, since for sufficiently large (constant) $n$, $\gamma n^{O(B)}\cdot 2^{4n\Rstar} \leq \frac12 n^{-A}$. Moreover, as we will shortly see, $|\DDrc^{\lambda}(\Rstar) - \DDrc^{s_j}(\Rstar)| \leq \poly(n) \cdot \|\lambda - s_j\|_2 \leq \poly(n) \cdot \gamma$, so, for a fixed codebook $\CCC$ and sufficiently large $n$,
    \begin{align*}
        \sup_{\Sigma_X = U \Lambda U^\top, U\in\mathcal{O}_n} \paren{\E_W \sqb{d_{\Sigma_X}(W,\Wh)} - n \DDrc^{\lambda}(\Rstar)} \leq \sup_{\Sigma_X = U_i \Lambda_j U_i^\top} \paren{\E_W \sqb{d_{\Sigma_X}(W,\Wh)} - n \DDrc^{s_j}(\Rstar)} + n^{-A}\,.
    \end{align*}
    Then, 
    $$
    \Pr_{\CCC|E_\CCC}\sqb{\sup_{\Sigma_X = U \Lambda U^\top, U\in\mathcal{O}_n} \paren{\E_W \sqb{\frac 1 n d_{\Sigma_X}(W,\Wh)} - \DDrc^{\lambda}(\Rstar)} \leq \eta + n^{-A-1}} \geq 1 - NN'\exp\paren{-2^{n\eps(1-c'n^{4\alpha - 1})}}\,.
    $$
    In the above, $NN' \leq (c/\gamma)^{n^2 + n} = 2^{5(n^3+n^2)\Rstar}$, so the probability on the RHS of the equation above is $\geq 1 - \exp\paren{-2^{n\eps(1-c n^{4\alpha - 1})}}$ for a constant $c$. The statement of Theorem~\ref{thm:union} follows by a union bound and the fact that $\Pr_\CCC\sqb{ E_\CCC^c} \leq  \exp\paren{- c' n^{2B}}$ and choosing $n$ such that $n^{-A-1} \leq \eta$.

    It remains to verify that $|\DDrc^{\lambda}(\Rstar) - \DDrc^{s_j}(\Rstar)| \leq \poly(n) \cdot \|\lambda - s_j\|_2$. An explicit calculation (using the implicit function theorem) gives for all $j\in \sqb{n}$
    \begin{align*}
        \frac{\partial n\DDrc^{\lambda}(\Rstar)}{\partial \lambda_j} &= \frac{1}{(1+\lambda_j T)^2} + \sum_i \frac{\lambda_i^2}{(1+\lambda_iT)^2} \cdot \frac{\frac{T}{1+\lambda_j T}}{\sum_i \frac{\lambda_i}{1+\lambda_i T}}\\
        &\leq \frac{1}{(1+\lambda_j T)^2} + \frac{1}{1+\lambda_j T} \leq 2\,,
    \end{align*}
    which yields the desired statement.

\begin{proof}(of Claim~\ref{claim:dist_perturb_union})
    Denote $\tau = \tau(W, \Sigma_X)$ and $\taut = \tau(W, \Sigmat_X)$. In our notation the claim statement is equivalent to, for $W\sim \N(0, I_n)$,
    $$
    \E_W\sqb{ d_{\Sigmat_X} (W, q(\taut) c_i)} \leq \E_W\sqb{ d_{\Sigma_X} (W, q(\tau) c_j)} + \frac12 n^{-A} + \gamma n^{O(B)} \cdot 2^{4n\Rstar}\,,
    $$
    where $i = \argmin_{k} d_{\Sigmat_X} (W, \taut c_k)$ and $j = \argmin_{k} d_{\Sigma_X} (W, \tau c_k)$. 
    
    \textbf{Error from quantizing $\tau, \taut$.} We first uniformly bound the effect of quantizing $\tau, \taut$. Similarly to Claim~\ref{claim:tau_quant_bound},
    \begin{equation*}
        \B|\sqrt{d_{\Sigma_X}(W, q(\tau) c_j)} - \sqrt{d_{\Sigma_X}(W, \tau c_j)}\B|  \leq \delta_\tau \sqrt{c_j^\top \Sigma_X c_j} \leq \delta_\tau n^{B+1}\,. 
    \end{equation*}
    A standard Gaussian tail argument, combined with a bound $\max\sset{q(\tau), \tau} \leq \|U^\top W\|_\infty$, shows that $\E_W \sqb{ \sqrt{d_{\Sigma_X}(W, q(\tau) c_j)} + \sqrt{d_{\Sigma_X}(W, \tau c_j) }} \leq n^{P}$ for a sufficiently large constant $P$, so setting $\delta_\tau \leq \frac18 n^{-B-1 - P - A}$ incurs $o_n(1)$ factors in $R$ (see proof of Theorem~\ref{thm:fixed_sigma}) and achieves 
    $$
    \B|d_{\Sigma_X}(W, q(\tau) c_j) - d_{\Sigma_X}(W, \tau c_j)\B| \leq \frac 18 n^{-A} \quad\text{and} \quad \B|d_{\Sigmat_X}(W, q(\taut) c_i) - d_{\Sigmat_X}(W, \taut c_i)\B| \leq \frac 18 n^{-A}\,.
    $$
    Thus, it is sufficient to show
    $$
    \E_W\sqb{ \min_i d_{\Sigmat_X} (W, \taut c_i)} \leq \E_W\sqb{ \min_i d_{\Sigma_X} (W, \tau c_i)} + \frac14 n^{-A} + \gamma n^{O(B)}\cdot 2^{4n\Rstar} \,.
    $$

    \textbf{Error from the tail $W$ event.}
    First, denote the event $E_W = \sset{\|W\|_2 \leq n^{B}}$. 
    We can expand
    \begin{align*}
        \E_W\sqb{ \min_i d_{\Sigma_X} (W, \tau c_i)} = \E_W\sqb{ \min_i d_{\Sigma_X} (W, \tau c_i) \cdot \I_{E_W}} +  \E_W\sqb{ \min_i d_{\Sigma_X} (W, \tau c_i) \cdot \I_{E_W^c}}\,,
    \end{align*}
    and, by definition of $d_{\Sigma_X} (W, \tau c_i)$, the second term is 
    \begin{align*}
    \E_W\sqb{ \min_i d_{\Sigma_X} (W, \tau c_i) \cdot \I_{E_W^c}} &= \E_W\sqb{ \min_i (W - \tau c_i)^\top \Sigma_X (W - \tau c_i) \cdot \I_{E_W^c}} \\
    &\leq \E_W \sqb{ n \|U^\top W\|^2_\infty (1 + \|U^\top c_i\|_\infty) \cdot \I_{E_W^c}} \\&\leq 2n^{21} \E_W \sqb{ \|U^\top W\|^2_\infty \cdot \I_{E_W^c}}\\ 
    &\leq 2n^{2B+1} \E_W \sqb{ \| W\|^2_2 \cdot \I_{E_W^c}} \\
    &\leq 2n^{2B+1} \int_{t = n^{B}}^\infty t e^{-c_\chi t} \leq c' n^{3B} e^{-c'n^{B}} \leq \frac14 n^{-A}
    \end{align*}
    for sufficiently large $n$,
    where we used $\tau \leq \|U^\top W\|_\infty$, see Claim~\ref{claim:tau_bound}, and $c_\chi, c'$ are universal constants. 

    \textbf{Error from $\Sigma_X$ perturbation.}
    We now show that \begin{equation}\label{eq:trunc_perturb}
        \E_W\sqb{ \min_i d_{\Sigmat_X} (W, \taut c_i) \cdot \I_{E_W}} \leq \E_W\sqb{ \min_i d_{\Sigma_X} (W, \tau c_i) \cdot \I_{E_W}} + \gamma n^{O(B)} \cdot 2^{4n\Rstar}\,,
    \end{equation}
    which, combined with the bound above yields the Claim. We use Lemma~\ref{lem:tau_perturb_union} (proved in Sec.~\ref{subsec:tau_perturb_union}), which bounds $\b|\tau(W, \Sigma_X) - \tau(W, \Sigmat_X)\b|$ in terms of $\| U^\top W - \Ut^\top W \|_\infty$ and $\|\Lambda - \Lambdat\|_2$. Given that $\|U^\top W\|_{\infty}, \|\Ut^\top W\|_\infty \leq \|W\|_2 \leq n^{B}$ under the $E_W$ event, Lemma~\ref{lem:tau_perturb_union} gives a (crude) bound:
    $$
    |\tau - \taut| = \b|\tau(W, \Sigma_X) - \tau(W, \Sigmat_X)\b| \leq \| U^\top W - \Ut^\top W \|_\infty + n^{O(B)} \|\Lambda - \Lambdat\|_2 \cdot 2^{4n\Rstar}\,.
    $$
    Plugging in $\| U^\top W - \Ut^\top W \|_\infty \leq \| U - \Ut\|_{op}\|W\|_{2} \leq \gamma n^{B}$ and $\|\Lambda - \Lambdat\|_2 \leq \gamma \sqrt{n}$, we obtain 
    $$
    |\tau - \taut| \leq \gamma n^{O(B)}\cdot 2^{4n\Rstar}\,.
    $$
    Additionally, 
    \begin{align*}
        \|\Sigmat_X - \Sigma_X \|_{op} &\leq \|U^\top (\Lambda - \Lambdat) U \|_{op} + \|U^\top \Lambdat U - \Ut^\top \Lambdat \Ut\|_{op}\\
        &\leq \|\Lambda - \Lambdat\|_2 + \|(U-\Ut)^\top \Lambdat \Ut\|_{op} + \| \Ut^\top \Lambdat (U - \Ut)\|_{op} + \| (U-\Ut)^\top \Lambdat (U - \Ut)\|_{op}\\
        &\leq \gamma\sqrt{n} + 2\gamma n + \gamma^2 n \leq 4 \gamma n\,.
    \end{align*}
    Again denote $j = \argmin_{k} d_{\Sigma_X} (W, \tau c_k)$, and let $x = W - \tau C_j$, $\xt = W - \taut C_j$. Then,
    \begin{align*}
        \min_i d_{\Sigmat_X} (W, \taut c_i) \cdot \I_{E_W} &\leq d_{\Sigmat_X} (W, \taut c_j) \cdot \I_{E_W}\\ 
        &= (W - \taut c_j)^\top \Sigmat_X (W- \taut c_j) \cdot \I_{E_W} \\
        &= \xt^\top \Sigmat_X \xt \cdot \I_{E_W}\\
        &= x^\top \Sigma_X x \cdot \I_{E_W} + (\xt-x)^\top \Sigmat_X \xt \cdot \I_{E_W} + x^\top \Sigma_X (\xt-x) \cdot \I_{E_W}\\
        &\qquad\qquad+ x^\top (\Sigmat_X - \Sigma_X) \xt \cdot \I_{E_W}\,.
    \end{align*}
    Since $\max \sset{\tau,\taut} \leq \max\sset{\|U^\top W\|_\infty, \|\Ut^\top W\|_\infty} \leq \|W\|_2$, we can bound 
    $$
    \|x\|_2 \cdot \I_{E_W}, \|\xt\|_2 \cdot \I_{E_W} \leq \|W\|_2 (1+ \|C_j\|_2) \cdot \I_{E_W} \leq 2 n^{2B}
    $$
    and, using the bound on $|\tau - \taut|$ above,
    \begin{align*}
        \|x - \xt\|_2 \leq |\tau - \taut| \cdot n^{B} \leq \gamma n^{O(B)} \cdot 2^{4n\Rstar} \,. 
    \end{align*}
    Plugging this into our bound, we obtain 
    \begin{align*}
        \min_i d_{\Sigmat_X} (W, \taut c_i) \cdot \I_{E_W} &\leq \min_i d_{\Sigma_X} (W, \tau c_i) \cdot \I_{E_W} + 2\gamma n^{O(B)} \cdot 2^{4n\Rstar} + 4n^{4B} \cdot 4 \gamma n\\
        &\leq \min_i d_{\Sigma_X} (W, \tau c_i) \cdot \I_{E_W} + \gamma n^{O(B)} \cdot 2^{4n\Rstar}\,,
    \end{align*}
    which concludes the proof of this step and the claim.
\end{proof}
\qed

\section{Worst-Case Gap Between Waterfilling and Random Coding: Theorem~\ref{thm:worst_case}}\label{sec:worst_case}

Recall that for $\Lambda = \diag\paren{\lambda_1,\dots,\lambda_n} \succeq 0$ with $ \trace (\Lambda) = n$, the random-coding rate-distortion is given by a parametric relationship 
\begin{equation}\label{eq:rc_rd_worst}
    \Drc(\Lambda, T) = \frac 1 n \sum_{i=1}^n \frac{\lambda_i}{1+\lambda_i T}\qquad \Rrc(\Lambda, T) = \frac 1 {2 n}\sum_{i=1}^n \log(1+\lambda_i T)\,,\tag{RDRC}
\end{equation}
and the waterfilling rate-distortion by 
\begin{equation}\label{eq:wf_rd_worst}
    \Dwf(\Lambda, t) = \frac{1}{n}\sum_{i=1}^n \min\sset{\lambda_i,t}\qquad \Rwf(\Lambda, t) = \frac 1 {2n} \sum_{i=1}^n \max\sset{0, \log (\lambda_i/t)}\,.\tag{RDWF}
\end{equation}
Throughout, $\log$ denotes the base-2 logarithm, and $\ln$ denotes the natural logarithm. The goal of this section is to quantify, at a fixed target distortion $\Dstar$, the rate overhead of our universal codebook, whose rate-distortion is given in Eq.~\eqref{eq:rc_rd_worst}, compared to the $\Sigma_X$-fine-tuned optimal codebook, whose rate-distortion is given in Eq.~\eqref{eq:wf_rd_worst}.

Denote\footnote{All of the $\Drc, \Dwf, \Rrc, \Rwf$ are monotone in $T, t$.} the implicit functions $\Trc(\Lambda, \Dstar) = T^\star$ s.t. $\Drc(\Lambda, T^\star) = \Dstar$ and $\twf(\Lambda, \Dstar) = t^\star$ s.t. $\Dwf(\Lambda, t^\star) = \Dstar$. Denote 
\begin{equation*}\label{eq:rrrc_rrwf}
    \RRrc(\Lambda, \Dstar) = \Rrc(\Lambda, \Trc(\Lambda, \Dstar)) \qquad\text{and}\qquad \RRwf(\Lambda, \Dstar) = \Rwf(\Lambda, \twf(\Lambda, \Dstar))\,.
\end{equation*}
In this notation, the rate overhead incurred by our universal codebook is 
$$
\sup_{\Lambda} \sset{ \RRrc(\Lambda, \Dstar) - \RRwf(\Lambda, \Dstar)}\,,
$$
where the supremum is over $\Lambda = \diag(\lambda_1,\dots,\lambda_n)\succeq 0$ with $\trace(\Lambda) = n$. In Theorem~\ref{thm:worst_case}, we prove
$$
\sup_{\Dstar \in (0,1)}\sup_{\Lambda} \sset{ \RRrc(\Lambda, \Dstar) - \RRwf(\Lambda, \Dstar)} \leq 0.11\,.
$$
Concretely, we first verify that the maximum rate gap is approached at spectra containing at most $2$ distinct eigenvalue and at vanishing distortions. The resulting gap expression can be directly bounded by $0.11$ bit. 

\subsection{Proof of Theorem~\ref{thm:worst_case}}\label{subsec:worst_case_proof}

\begin{mythm}{\ref{thm:worst_case}}\textnormal{\textbf{(Worst-Case Gap Between Waterfilling and Random Coding)}}
    $$
    \sup_{\Dstar \in (0,1)}\sup_{\Lambda} \sset{ \RRrc(\Lambda, \Dstar) - \RRwf(\Lambda, \Dstar)} \leq 0.11\,,
    $$
    where the supremum is over $\Lambda = \diag(\lambda_1,\dots,\lambda_n)\succeq 0$ with $\trace(\Lambda) = n$.
\end{mythm}

Theorem~\ref{thm:worst_case} follows immediately from a more general version of the same statement in Theorem~\ref{thm:worst_case_measure}. Let $\mathcal{P}([0,\infty))$ denote the space of Borel probability measures on $[0,\infty)$, and define
$$
\mathcal{P}_1([0,\infty)) = \sset{\lambda \in \mathcal{P}([0,\infty)): \int x d\lambda(x) = 1}\,.
$$
We extend the finite-dimension rate-distortion curves in Eq.~\eqref{eq:rc_rd_worst} and \eqref{eq:wf_rd_worst} as follows: for $\mu \in \mathcal{P}_1([0,\infty))$, let
$$
\Drc(\mu,T)
=
\int_0^\infty \frac{x}{1+xT}d\mu(x),
\qquad
\Rrc(\mu,T)
=
\frac12\int_0^\infty \log(1+xT)d\mu(x),
$$
and
$$
\Dwf(\mu,t)
=
\int_0^\infty \min\sset{x,t}d\mu(x),
\qquad
\Rwf(\mu,t)
=
\frac12\int_0^\infty \max\sset{0,\log(x/t)}d\mu(x).
$$
Similarly to above, for each $\Dstar\in(0,1)$, let $\Trc(\mu,\Dstar)$ and $\twf(\mu,\Dstar)$ be defined by
$
\Drc(\mu,\Trc(\mu,\Dstar))=\Dstar$ and $\Dwf(\mu,\twf(\mu,\Dstar))=\Dstar,
$
and set
$$
\RRrc(\mu,\Dstar)
=
\Rrc(\mu,\Trc(\mu,\Dstar)),
\qquad
\RRwf(\mu,\Dstar)
=
\Rwf(\mu,\twf(\mu,\Dstar)).
$$

\begin{theorem}[Worst-Case Gap Between Waterfilling and Random Coding]\label{thm:worst_case_measure} 
    $$
   \sup_{\Dstar \in (0,1)} \sup_{\mu} \sset{ \RRrc(\mu, \Dstar) - \RRwf(\mu, \Dstar)} \leq 0.11\,,
    $$
    where the supremum is over probability measures $\mu \in \mathcal{P}_1([0, \infty))$.
\end{theorem}

Indeed, for any $\Lambda=\diag(\lambda_1,\dots,\lambda_n)\succeq 0$ with $\trace(\Lambda)=n$, the empirical measure $\mu_\Lambda=\frac1n\sum_{i=1}^n\delta_{\lambda_i}$ belongs to $\mathcal P_1([0,\infty))$.

\begin{proof}
Denote the rate-gap
$$
\Delta (\mu, \Dstar) = \frac 12 \int r_{\Trc(\mu, \Dstar), \twf(\mu, \Dstar)}(x) d \mu(x)\,, \qquad r_{T, t}(x) = \begin{cases}
    \log(1+xT) &\text{if } x \leq t\\
    \log\paren{\frac{1+xT}{x/t}} &\text{if } x > t\,.
\end{cases}
$$
It will be convenient to consider the following rescaling:
\begin{equation}\label{eq:rescaling}
\mut \coloneq (x \to x/\Dstar)_{\sharp \mu}\,,\tag{$\mut$}
\end{equation}
so that $\int x d \mut (x) = (\Dstar)^{-1} > 1$. Accordingly, we define $\Dstar(\mut) = \paren{\int x d \mut (x)}^{-1}$. Additionally, rescale the implicitly defined parameters $\twf$ and $\Trc$:
$$
\ttildewf(\mut) \triangleq \twf\paren{ (x \to x\cdot\Dstar)_{\sharp \mut} , \Dstar(\mut)} / \Dstar(\mut) \qquad \text{and} \qquad \Ttilderc(\mut) \triangleq \Trc\paren{ (x \to x\cdot\Dstar)_{\sharp \mut} , \Dstar(\mut)} \cdot \Dstar(\mut)\,,
$$
i.e., $\ttildewf(\mut) = \twf(\mu, \Dstar) / \Dstar$ and $\Ttilderc(\mut) = \Trc(\mu, \Dstar)\cdot\Dstar$ in the original parameters. The distortion conditions then correspond to 
\begin{align}
    1 = \int \frac{x}{1+x \Ttilde} d \mut(x) = \int \min\sset{x,\ttilde} d\mut(x)\,, \quad \text{where}\quad\ttilde = \ttildewf(\mut), \Ttilde = \Ttilderc(\mut)\,.\tag{Dist}\label{eq:dist_reparam}
\end{align}
We can now rewrite the optimization objective in terms of $\mut$:
\begin{equation}\label{eq:objective_reparam}
    \Phi(\mut) \triangleq \Delta\paren{ (x \to x\cdot\Dstar)_{\sharp \mut} , \Dstar(\mut) } = \Delta(\mu, \Dstar) = \frac12 \int r_{\Ttilde, \ttilde}(x) d\mut(x) \tag{$\Phi$}\,.
\end{equation}
Finally, for technical reasons, we will consider $\mut \in \mathcal{P}([0, \infty]) \supseteq \mathcal{P}([0, \infty))$ with the natural extensions of functions $\frac x {1+x \Ttilde}, \log\paren{\frac{1+x\Ttilde}{x/\ttilde}}$. To prove the theorem it is then sufficient to show
\begin{align*}
    \sup_{\mut \in \calA} \Phi(\mut) \leq 0.11\,, \quad \text{where}\quad\calA \triangleq \Big\{ &\mut \in  \mathcal{P}([0, \infty]): \Dstar(\mut) \in (0,1)\Big\}\,.
\end{align*}
In what follows, we will restrict the optimization scope $\calA$ in three steps. For $\mut$ as above, to simplify notation, when clear from the context, we write $\Dstar = \Dstar(\mut)$, $\ttilde = \ttildewf(\mut)$, $\Ttilde = \Ttilderc(\mut)$.

\noindent\textbf{Step I. } Denote, for $c = 0.11 \ln 2$, 
$$
\calA_{\mathrm{bd}} = \sset{\mut \in \calA:\ \ttildewf(\mut) \in \sqb{1, c^{-1} }, \Ttilderc(\mut) \in \sqb{ c, 1},\text{ and } \ttildewf(\mut)\cdot \Ttilderc(\mut) \leq 1}\,.
$$
As a first step, we show 
$$
\sup_{\mut \in \calA} \Phi(\mut) \leq \max\sset{0.11, \sup_{\mut \in \calA_{\mathrm{bd}}} \Phi(\mut)}\,.
$$
Let $\kappa = \kappa(\mut) \triangleq \int_{\ttilde}^{+\infty} d\mut(x)$ be the mass $\mut$ puts to the waterfilling-active modes.\footnote{$\kappa > 0$ from the water-filling distortion condition.} Expanding $\int \min\sset{x,\ttilde} d\mut(x) = \int_{0}^{\ttilde}x d\mut(x) + \kappa \ttilde$, from Eq.~\eqref{eq:dist_reparam} we conclude $\ttilde \geq 1$, $\Ttilde \leq 1$, and $\ttilde \Ttilde \leq 1$, where the last inequality is derived as follows:
   \begin{align*}
       1 = \int \frac{x}{1+x \Ttilde} d \mut(x) \leq \int_{0}^{\ttilde} x d \mut(x) + \frac \kappa \Ttilde = (1 - \kappa \ttilde) + \frac \kappa \Ttilde = 1 + \kappa \paren{-\ttilde + \frac 1 \Ttilde}\,.
   \end{align*}
   We now obtain bounds on $\Phi(\mut)$ in terms of $\ttilde, \Ttilde$:
   \begin{align*}
       \Phi(\mut) &= \underbrace{\frac 1 {2} \int_0^{\ttilde} \log(1+x \Ttilde) d \mut(x)}_{\mathrm{I}} +  \underbrace{\frac 1 {2} \int_{\ttilde}^{+\infty} \log \paren{\frac{1+x \Ttilde}{x / \ttilde}}d\mut(x)}_{\mathrm{II}}\,.
   \end{align*}
   From $\int \min \sset{x, \ttilde} d\mut(x) = 1$ and $\log(1+x) \leq x / \ln 2$, we bound 
   $\mathrm{I} \leq \frac {(1-\kappa)\Ttilde} {2\ln 2}$. To bound $\mathrm{II}$, notice that for $x \geq \ttilde$, $\log \paren{\frac{1+x\Ttilde}{x / \ttilde}} = \log (\ttilde\Ttilde + \ttilde /x) \leq \log(\ttilde\Ttilde + 1) \leq \ttilde \Ttilde /\ln 2$, so overall, 
   \begin{equation}\label{eq:gap_bound_tT}
   \mathrm{I} + \mathrm{II} \leq \frac {(1-\kappa)\Ttilde} {2\ln 2} + \frac{\kappa\ttilde \Ttilde}{2\ln 2}  \leq \frac {\Ttilde} {2\ln 2} + \frac{\Ttilde}{2\ln 2} = \frac {\Ttilde} {\ln 2} \leq \frac 1 {\ttilde \ln2}\,,
   \end{equation}
   where the last inequality used previously derived $\ttilde\Ttilde \leq 1$. Then, if $\ttilde \geq c^{-1}$ or $\Ttilde \leq c$, we obtain the desired bound $\Phi(\mut) \leq 0.11$, which concludes this step. 

   \noindent\textbf{Step II. } Second, we show that the largest gap is obtained at spectra, whose mass in the active tail (i.e., values $\geq \ttilde$) is distributed between endpoints $\ttilde$ and $+\infty$. Concretely, we show that 
   $$
    \sup_{\mut \in \calA_{\mathrm{bd}}} \Phi(\mut) \leq \sup_{\mut \in \calA_{\mathrm{tail}}} \Phi(\mut)\,,
    $$
    where 
    $$
    \calA_{\mathrm{tail}} = \sset{\mut \in \calA_{\mathrm{bd}}:\ \text{for } x\geq \ttildewf(\mut), \theta \in [0,1], \ \mut(x) \propto \theta \cdot \delta_{\ttilde}(x) + (1-\theta) \cdot \delta_{+\infty}(x)}\,.
    $$
    Recall the contribution to the rate gap of the active modes (those $\geq \ttilde$) is $$
    \frac 1 {2} \int_{\ttilde}^{+\infty} \log \paren{\frac{1+x \Ttilde}{x / \ttilde}}d\mut(x) = \frac 1 {2} \int_{\ttilde}^{+\infty} \log \paren{\ttilde \cdot \paren{\frac{x}{1+x \Ttilde}}^{-1}}d\mut(x)\,.
    $$
    The distortion constraints in Eq.~\eqref{eq:dist_reparam} yield an equality $\int_{\ttilde}^{\infty}\frac {x}{1+x \Ttilde} d\mut(x) = 1 - \int_{0}^{\ttilde}\frac {x}{1+x \Ttilde} d\mut(x)$. Note that for $x \in [\ttilde, +\infty]$, $\frac {x}{1+x \Ttilde} \in \sqb{\frac {\ttilde}{1+\ttilde \Ttilde}, \frac 1 \Ttilde}$ and is a monotone function. Moreover, both the objective $f(y) = \log (\ttilde y^{-1})$ and the first moment constraint are convex, so the supremum is obtained at $$\mut(x) \propto \theta \cdot \delta_{\ttilde}(x) + (1-\theta) \cdot \delta_{+\infty}(x)\qquad\text{for }x\geq \ttilde$$ for some $\theta \in [0,1]$. Note that this suggests that the extremal regime occurs at vanishing distortions.

    \noindent\textbf{Step III. } Finally, we show that the inactive modes (those $\leq \ttilde$) of the worst-case spectra $\mut$ equalize, i.e., 
    $$
    \sup_{\mut \in \calA_{\mathrm{tail}}} \Phi(\mut) \leq \sup_{\mut \in \calA_{\mathrm{2pt}}} \Phi(\mut)\,,
    $$
    where 
    $$
    \calA_{\mathrm{2pt}} = \sset{\mut \in \calA_{\mathrm{tail}}:\ \text{for }\theta \in [0,1], \mut_0 \leq \ttildewf(\mut), \ \mut(x) \propto \theta \cdot \delta_{\mut_0}(x) + (1-\theta) \cdot \delta_{+\infty}(x)}\,.
    $$
    For a given $\mut \in \calA_{\mathrm{tail}}$ with $\mut \propto (1-\theta)\cdot \delta_{+\infty}(x)$ for $x > \ttilde$, define a corresponding measure $\nu \in \calA_{\mathrm{2pt}}$ as 
    $$
    \nu \propto \theta \cdot \delta_{\mut_0} + (1-\theta)\cdot \delta_{+\infty}(x)\,, \quad \text{where} \quad \mut_0 = \frac{\int_{0}^{\ttilde} x d \mut(x)} \theta\,. 
    $$
    To complete Step III, it is sufficient to show that $\Phi(\mut) \leq \Phi(\nu)$. Define a path $\mu_s = s \cdot \mut + (1-s) \cdot \nu$ for $s \in [0, 1]$, and $J(s) \triangleq \Phi(\mu_s)$. We will show that $\frac {\partial J(s)} {\partial s} \leq 0$ for all $s\in \sqb{0,1}$, yielding $\Phi(\mut) = J(1) \leq J(0) = \Phi(\nu)$.
    
    For this, we compute:
    \begin{align*}
        \frac{\partial \Ttilderc(\mu_s)}{\partial s} &= -\frac{\frac{\partial \int \frac {x}{1+x\Ttilde} d \mu_s(x) }{\partial s}}{\frac{\partial \int \frac {x}{1+x\Ttilde} d \mu_s(x) }{\partial \Ttilde}} = \frac{ \int \frac {x}{1+x\Ttilde} d (\mut(x) - \nu(x)) }{ \int \frac {x^2}{(1+x\Ttilde)^2} d \mu_s(x) }\,.
    \end{align*}
    Note that for all $s \in \sqb{0,1}$, $\ttildewf(\mu_s)$ is unchanged. Then,
    \begin{align*}
        \frac{\partial 2J(s)}{\partial s} &= \frac { \partial \paren{\int_0^{\ttilde} \log(1+x\Ttilde) d \mu_s(x) + (1-\theta) \log (\Ttilde \ttilde) }} {\partial s}\\
        &=\int_0^{\ttilde} \log(1+x\Ttilde) d (\mut(x) - \nu(x)) + \frac 1 {\ln 2} \cdot \paren{\int_0^{\ttilde} \frac{x}{1+x\Ttilde} d \mu_s(x)+  \frac {1-\theta} {\Ttilde} } \cdot \frac{ \int \frac {x}{1+x\Ttilde} d (\mut(x) - \nu(x)) }{ \int \frac {x^2}{(1+x\Ttilde)^2} d \mu_s(x) }\\
        &= \int_0^{\ttilde} \paren{\log(1+x\Ttilde) + A_s \cdot \frac {x}{1+x\Ttilde} } d (\mut(x) - \nu(x)) \,,
    \end{align*}
    where $A_s \geq 0$. The function $\log(1+x\Ttilde) + A_s \cdot \frac {x}{1+x\Ttilde}$ is concave on $x \in [0, \ttilde]$, yielding the final inequality $\frac{\partial J(s)}{\partial s} \leq 0$.
    \medskip
    It remains to show $\sup_{\mut \in \calA_{\mathrm{2pt}}} \Phi(\mut) \leq 0.11$. Recall that every $\mut \in \calA_{\mathrm{2pt}}$ can be expressed for some $\mut_0 \leq \ttildewf(\mut)$ as $$
    \mut(x) = \theta \cdot \delta_{\mut_0}(x) + (1-\theta) \cdot \delta_{+\infty}(x)\,.
    $$
    We can then express
    $$
    2\Phi(\mut) =  \theta \log(1+\mut_0\Ttilde) + (1-\theta) \log (\ttilde \Ttilde)\,,
    $$
    where from Eq.~\eqref{eq:dist_reparam},
    $$
    1 = \theta\mut_0 + (1-\theta)\ttilde = \frac{\theta\mut_0}{1+\mut_0\Ttilde} + \frac{1-\theta}{\Ttilde}\,.
    $$
    The case $\theta = 1$ is impossible for $\Ttilde > 0$, so we will focus on the $\theta \in [0, 1)$ case.
    We now further simplify this optimization problem to obtain the final $0.11$ bound. 
    
    From the first distortion constraint, we express $\ttilde = (1-\theta)^{-1} - \paren{(1-\theta)^{-1} - 1} \mut_0$. From the second distortion constraint, $\Ttilde = \frac{\theta\mut_0\Ttilde}{1+\mut_0\Ttilde} + 1-\theta$. Denoting $\alpha = \mut_0 \Ttilde$, we then simplify $\ttilde\Ttilde  = (1-\theta)^{-1}\Ttilde - \paren{(1-\theta)^{-1} - 1} \alpha$. Plugging into the rate-gap expression:
    \begin{align*}
        2\Phi(\mut) &\leq \theta\log(1+\alpha) + (1-\theta) \log \paren{1 + \frac{\theta\alpha}{(1+\alpha)(1-\theta)} - \frac{\alpha\theta}{1-\theta}}\\
        &=\theta\log(1+\alpha) + (1-\theta) \log \paren{1 - \frac{\alpha^2 \theta}{(1+\alpha)(1-\theta)} }\,,
    \end{align*}
    where from the positivity of $\ttilde, \Ttilde$ we have $0 \leq \theta < \frac{1+\alpha}{1+\alpha+\alpha^2}$ and $0 \leq \alpha = \mut_0 \Ttilde \leq \ttilde \Ttilde \leq 1$. Denoting $\rho = \frac{\theta}{1-\theta}$ and differentiating with respect to $\alpha$ gives 
    \begin{align*}
        \ln 2 (1-\theta)^{-1}\frac{\partial (2\Phi(\mut))}{\partial \alpha} &= \frac{\rho}{1+\alpha} - \frac{1}{1 - \frac{\alpha^2\rho}{1+\alpha}}\cdot \rho \cdot \frac{2\alpha + \alpha^2}{(1+\alpha)^2}\\
        &=\frac{\rho}{1+\alpha}\paren{1 - \frac{2\alpha+\alpha^2}{1+\alpha - \alpha^2 \rho}}\,.
    \end{align*}
    Setting this to $0$ yields $\rho = \frac{1-\alpha - \alpha^2}{\alpha^2}$. At the boundary points $\alpha \in \sset{0,\frac{\sqrt{5} - 1}{2}}$, the objective function is below the objective of the Theorem. Plugging into the objective, we derive a one-parameter function 
    \begin{align*}
        2\Delta^\star \leq \frac{1-\alpha - \alpha^2}{1-\alpha} \log(1+\alpha) + \frac{\alpha^2}{1-\alpha}\log\paren{1 - \frac{1-\alpha-\alpha^2}{1+\alpha}}\,.
    \end{align*}
    The maximum of this objective is achieved at $\alpha\approx 0.35$ and is $\approx 2\cdot 0.108 < 2\cdot 0.11$, which concludes the proof.

\end{proof}

\section{Proof of Lemma~\ref{lem:tau_perturb_union}}\label{subsec:tau_perturb_union}

\begin{lemma}[$\tau$ Perturbation]\label{lem:tau_perturb_union}
    Let $\Rstar$ be a fixed constant, $C_W = C_W(n) \in\R$, $W \in \R^n, s \in \R_{+}^n$, such that $\|W\|_\infty \leq C_W$ and $\sum_j s_j = n$. Denote
    \begin{equation*}
        \tau(W, s) = \paren{T \sum_{j} \frac{W_j^2 s_j^2}{(1+s_j T)^2}}^{1/2}\paren{\sum_j \frac{s_j}{1+s_j T}}^{-1/2}\,,
        \end{equation*}
    where $T = T(s)$ is a unique solution to $\Rrc^{s}(T) = \Rstar$ (see Eq.~\eqref{eq:rd_rc} for definition of $\Rrc^{s}(T)$).
    If $W, \Wh, s, \sh \in \R^n$ satisfy $\|W\|_\infty, \|\Wh\|_\infty \leq C_W$ and $\sum_j s_j = \sum_j \sh_j = n$ and 
    $$
    \|W - \Wh\|_\infty \leq \gamma_1 C_W \quad\text{and}\quad \|s - \sh\|_2 \leq \gamma_2 \sqrt{n}\,,
    $$
    then 
    $$
    \b|\tau(W, s) - \tau(\Wh, \sh) \b| \leq \gamma_1 C_W + c \cdot C_W \gamma_2 n^4 2^{4n\Rstar}
    $$
    for sufficiently large $n$ and a constant $c = c(\Rstar)$.
\end{lemma}

\paragraph{Proof of Lemma~\ref{lem:tau_perturb_union}.}
    Denote $T,\Th$ to be the solutions to $\Rrc^{s}(T) = \Rstar$ and $\Rrc^{\sh}(\Th) = \Rstar$ respectively, where recall that $\Rrc^{V, s} = \Rrc^s$ does not depend on $V$. 
    \paragraph{Step 1: Bound on $|T - \Th|$.} First, we show that under the Lemma conditions, 
    \begin{equation}\label{eq:T_bound}
        |T - \Th| \leq  \gamma_2 n \cdot 2^{4n\Rstar}\,.
    \end{equation}
    First, since for each $j\in \sqb{n}$, $\b|\log(1+s_j T) - \log(1+\sh_j T)\b| \leq \frac 1 {\ln 2}\cdot \frac{T}{1+\min(s_j,\sh_j)T} \cdot |s_j - \sh_j| \leq T|s_j - \sh_j| / \ln 2 $, we have
    $$
    \b| \underbrace{\Rrc^{ {\sh}}(\Th)}_{\Rstar} - \Rrc^{{\sh}}(T) \b| = \b| {\Rrc^{s}(T)} - \Rrc^{{\sh}}(T) \b|\leq \frac{T \|s - \sh\|_1} {2n\ln 2} \leq \frac{T \|s - \sh\|_2} {2\sqrt{n}\ln 2} \leq \frac{T\gamma_2}{2\ln2}\,.
    $$
    By the mean value theorem, we have for some $\bar T \in \sqb{\min(T,\Th), \max(T,\Th)}$,
    $$
    |T - \Th| = \frac{\b|\Rrc^{ {\sh}}(\Th) - \Rrc^{{\sh}}(T) \b|}{(\Rrc^{{\sh}})'(\bar T)} \leq \frac{T\gamma_2}{2\ln2 \cdot (\Rrc^{{\sh}})'(\bar T)}\,.
    $$
    A direct calculation gives $(\Rrc^{{\sh}})'(\bar T) = \frac1{2n\ln 2} \sum_j \frac{\sh_j}{1 + \sh_j \bar T} \geq \frac{1}{2n\ln 2(1+\bar T)} \geq \frac{1}{2n\ln 2(1+ \max (T, \Th))}$, where we used that $\max_j \sh_j \geq 1$. Moreover, for any $s \in \R_{+}^n$ with $\|s\|_1 = n$, $2 n \Rrc^{{s}}(T) \geq \log(1+T)$, and therefore, $1 + \max(T, \Th) \leq 2^{2n\Rstar}$. This yields 
    $$|T - \Th| \leq \gamma_2 n \cdot T (1 +  \max(T, \Th)) \leq \gamma_2 n \cdot 2^{4n\Rstar}\,.$$

    \paragraph{Step 2: Bound on $|\tau(W, s) - \tau(\Wh, s)|$.} Here we show that 
    \begin{equation}\label{eq:W_bound}
        |\tau(W, s) - \tau(\Wh, s)| \leq \gamma_1 C_W\,.
    \end{equation}
    By triangle inequality and $\frac{s_j}{1+s_j T} \leq \frac 1 {T}$, 
    \begin{align*}
    |\tau(W, s) - \tau(\Wh, s)| &\leq \paren{T \sum_{j} \frac{(W_j - \Wh_j)^2 s_j^2}{(1+s_j T)^2}}^{1/2}\paren{\sum_j \frac{s_j}{1+s_j T}}^{-1/2}\\
    &\leq \gamma_1 C_W \cdot \paren{T \sum_{j} \frac{s_j^2}{(1+s_j T)^2}}^{1/2}\paren{\sum_j \frac{s_j}{1+s_j T}}^{-1/2} \leq \gamma_1 C_W\,.
    \end{align*}

    \paragraph{Step 3: Bound on $|\tau(W, s) - \tau(W, \sh)|$.} Here we bound 
    $|\tau(W, s) - \tau(W, \sh)|$.
    For convenience, denote $d_j = \frac{s_j}{1+s_j T}$ $A (s, T) \eqcolon \sum_{j} W_j^2 d_j^2 $, $B(s, T) \eqcolon \sum_j d_j$ and let $A = A(s, T), \Ah = A(\sh, \Th)$ and $B = B(s, T), \Bh = B(\sh, \Th)$, so that
    $$
    |\tau(W, s) - \tau(W, \sh)| = \B|\frac{\sqrt{TA}}{\sqrt{B}} - \frac{\sqrt{\Th\Ah}}{\sqrt{\Bh}}\B| \leq \sqrt{T} \B| \frac{\sqrt{A}}{\sqrt{B}} - \frac{\sqrt{\Ah}}{\sqrt{\Bh}}\B| + \b|\sqrt{T} - \sqrt{\Th}\b|\cdot \frac{\sqrt{\Ah}}{\sqrt{\Bh}} \eqcolon \one + \two
    $$
    by triangle inequality. To bound $\one$, we use a triangle inequality $\one = \sqrt{T} \B| \frac{\sqrt{A}}{\sqrt{B}} - \frac{\sqrt{\Ah}}{\sqrt{\Bh}}\B| \leq \sqrt{T}\frac{\b|\sqrt{A} -\sqrt{\Ah}\b|}{\sqrt{B}} + \sqrt{T\Ah} \B|\frac{1}{\sqrt{B}} - \frac 1 {\sqrt{\Bh}}\B|$ and notice $$\b|\sqrt{A} - \sqrt{\Ah}\b| = \b| \|W\odot d\|_2 - \|W\odot \dhat\|_2 \b| \leq \|W\|_\infty\| d - \dhat \|_2 \leq C_W \|d - \dhat\|_2\,.$$ Moreover, 
    $$
    \b| \sqrt{B} - \sqrt{\Bh} \b| = \b| \|d\|^{1/2}_1 - \|\dhat\|^{1/2}_1 \b| \leq \frac{\sqrt{n} \| d - \dhat\|_2}{\sqrt{B} + \sqrt{\Bh}}\,.
    $$
    For a uniform lower bound on $B, \Bh$ we use $\min(B, \Bh) \geq \frac 1 {1+T}$, so the combined bound is 
    \begin{align*}
        \one &\leq \sqrt{T} \| d - \dhat \|_2 \sqb{C_W (1+T)^{1/2} + \sqrt{\Ah n } (1+T)^{3/2}}\\
        &\leq 2 \| d - \dhat \|_2 C_W n 2^{2n\Rstar}\,.
    \end{align*}
    where we used $\Ah \leq C_W^2 \frac n {T^2}$ and $1+T \leq 2^{2n\Rstar}$. It remains to bound $\| d - \dhat \|_2$. Denoting $d(s, T) = \paren{\frac{s_j}{1+s_jT}}_j$ (so that $d = d(s, T)$ and $\dhat = d(\sh, \Th)$), we have
    \begin{align*}
        \| d - \dhat \|_2 \leq \| d(s, T) - d(\sh, T)\|_2 + \|d(\sh, T) - d(\sh, \Th)\|_2\,.
    \end{align*}
    Since $\B|\frac{s_j}{1+s_j T} - \frac{\sh_j}{1+\sh_j T}\B| \leq |s_j - \sh_j|$, we have $\| d(s, T) - d(\sh, T)\|_2 \leq \| d(s, T) - d(\sh, T)\|_1 \leq \|s - \sh\|_1 \leq \gamma_2 n$. Moreover, $\B|\frac{s_j}{1+s_j T} - \frac{s_j}{1+s_j \Th}\B| \leq n^2 |T - \Th|$, so we obtain 
    $$
    \| d - \dhat \|_2 \leq \gamma_2 n + 2\gamma_2 n^4 2^{4n\Rstar} \leq 3\gamma_2 n^4 2^{4n\Rstar}\,, 
    $$
where we used $|T - \Th| \leq 2\gamma_2 n 2^{4n\Rstar}$ in Eq.~\eqref{eq:T_bound}.
    To bound $\two$, we use the bound in Eq.~\eqref{eq:T_bound} and the lower bound $T \geq 2\Rstar\ln 2$ obtained in the proof of Thm.~\ref{thm:fixed_sigma}:
    $$
    \b| \sqrt{T} - \sqrt{\Th} \b| = \frac{\b|T-\Th\b|}{\sqrt{T} + \sqrt{\Th}} \leq \frac{\gamma_2 n 2^{4n\Rstar}}{\sqrt{2\Rstar \ln 2}}\,.
    $$
    Then,
    $$
    \two \leq\frac{\gamma_2 n 2^{4n\Rstar}}{\sqrt{2\Rstar \ln 2}} \cdot \frac{C_W\sqrt{n}}{2T} \cdot \sqrt{1 + T} \leq \gamma_2 C_W n^{3/2} 2^{4n\Rstar} / (2\Rstar\ln2)\,.
    $$
    We have for the final bound 
    \begin{align*}
        |\tau(W, s) - \tau(W, \sh)| \leq \one + \two &\leq 6\gamma_2 n^4 2^{4n\Rstar} C_W n 2^{2n\Rstar} + \gamma_2 C_W n^{3/2} 2^{4n\Rstar} / (2\Rstar)\\
        &\leq c \cdot C_W \gamma_2 n^4 2^{4n\Rstar} 
    \end{align*}
    for sufficiently large $n$ and a constant $c = c (\Rstar)$. Analogous bound holds for $|\tau(\Wh, s) - \tau(\Wh, \sh)|$.

    \textbf{Step 4: Final bound.} Combining Steps 2,3 above, we obtain
    \begin{align*}
        \b| \tau(W,s) - \tau(\Wh,\sh)\b| &\leq \b| \tau(W,s) - \tau(\Wh,s)\b| + \b| \tau(\Wh,s) - \tau(\Wh,\sh)\b|\\
        &\leq \gamma_1 C_W + c \cdot C_W \gamma_2 n^4 2^{4n\Rstar}\,.
    \end{align*}
    
\qed

\section{Proof of upper bound in oracle Proposition~\ref{prop:oracle_wf}}\label{sec:ub_wf}

Here we sketch proof of~\eqref{eq:ub_wf}. First, since $f$ and $g$ are allowed to depend on
$\Sigma_X$ we can rotate $\Sigma_X$ to eigenbasis, and thus assume from now on that $\Sigma_X =
\diag(\lambda_1,\ldots,\lambda_n)$. Fix $t$ and let $D_i = \min\{t/\lambda_i, 1\}$. Let
$R_0=\Rwf(\Sigma_X,t)$. Fix arbitrary $\epsilon>0$ and set rate $R=R_0+2\epsilon$. We will show
that by generating codebook $\CCC$ randomly via sampling $1+2^{nR}$ codewords from distribution $\prod_{i=1}^n\mathcal{N}(0, 1-D_i)$ one can
attain distortion
\begin{equation}\label{eq:uwf_1}
	\EE_W \left[ \min_{c\in\CCC} d_{\Sigma_X}(W, c) \right] \le n \Dwf(\Sigma_X, t) +
\left(
e^{-2^{n\epsilon}} + c_1 e^{-c_2 n \epsilon^2} \right) \tr \Sigma_X\,,
\end{equation}
where $c_1,c_2>0$ are some absolute constants. From here the statement of the theorem follows by
setting $\epsilon_n = c\sqrt{\frac{\log n} n}$ with an apropriate $c>0$.

To show~\eqref{eq:uwf_1} let $Y_i \simiid \mathcal{N}(0, 1-D_i)$ independently. Let also $Z_i
\simiid \mathcal{N}(0,1)$. Set
$$ W_i = Y_i + \sqrt{D_i} Z_i $$
and notice that $W=(W_1,\ldots,W_n)\sim \mathcal{N}(0,I_n)$. This coupling of $W$ to $Y$ satisfies
a useful property:
$$ \EE_W[d_{\Sigma_X}(W,Y)] = \EE_W\left[\sum_{i=1}^n \lambda_i D_i Z_i^2\right] =
\Dwf(\Sigma_X,t)\,.$$

On the other hand, we have \textit{information density} 
$$ i(W;Y) := \log {\frac{dP_{W,Y}} {d(P_W \times P_Y)}}(W,Y) = n R_0 + {\frac{\log e } 2} \sum_{i=1}^n
W_i^2 - Z_i^2\,.$$
Note that $\EE[i(W;Y)] = I(W;Y) = nR_0$. Note that while $W_i,Z_i$ are jointly correlated
Gaussians, they are independent for different $i$'s. A fascinating property of information density
(underlying information stability) of Gaussian processes is that its variance is uniformly bounded
regardless of the distribution, see~\cite[(19.32)]{PWbook24}. We will exploit this below to show
``local'' subgaussian estimate on the concentration of $i(W;Y)$.

Define log-MGF function
$$ f(z) := {\frac 1 n} \ln \EE_{W,Z}\left[ e^{z  \sum_{i=1}^n W_i^2-Z_i^2}\right] = 
{\frac 1 n} \sum_{i=1}^n \ln \EE_{W_i,Z_i}\left[ e^{z (W_i^2-Z_i^2)}\right]\,.$$
It is not hard to show that there exists a neighborhood $\mathcal{S}$ of $0$ on the
complex plane $\mathbb{C}$ and a constant $c$ such that 
$$ \sup_{z\in\mathcal{S}}|f(z)| \le c\,,$$
and crucially $\mathcal{S}$ and $c$ can be chosen independent of $\{D_i\}$ (and hence
of $t,\lambda_i$). For this, one only needs to apply Cauchy-Scwharz to reduce to analysis of
log-MGF of $W_i^2$ and $Z_i^2$, which are just squares of $\mathcal{N}(0,1)$. 

In addition to being analytic, $f$ also satisfies $f(0)=f'(0)=0$ and $f''(0) \le 2 \EE[W_i^4 + Z_i^4] \le 12$. Thus,
by Cauchy formula we can uniformly bound $f''(z)$ inside any compact subset of $\mathcal{S}$.
Consequently, there exists a universal $c'>0$ and $z_0>0$ such that for all real $-z_0 < z < z_0$
we have 
$$ f(z) \le 2 c' z^2\,.$$

Applying Chernoff estimate we find that for all $\epsilon < \epsilon_0$  we have 
\begin{equation}\label{eq:uwf_2}
	\PP[i(W;Y) > I(W;Y) + n\epsilon] \le e^{-2c_2 n\epsilon^2}\,,
\end{equation}
where crucially $\epsilon_0,c_2$ are absolute constants.

We are now ready to apply standard finite blocklength rate-distortion upper bound~\cite[Theorem
25.2]{PWbook24}, which claims existence of codebook $\CCC$ with distortion
\begin{equation}\label{eq:uwf_3}
	\EE_W \left[ \min_{c\in\CCC} d_{\Sigma_X}(W, c) \right] \le \EE[ d_{\Sigma_X}(W,Y)] +
\EE[d_{\Sigma_X}(W,0)] e^{-2^{nR}/\gamma}  + \EE[ d_{\Sigma_X}(W,Y) 1\{i(W;Y) > \log_2
\gamma\}]\,,
\end{equation}
where $\gamma$ is arbitrary, but we set it to
$$ \log_2 \gamma = nR_0 + n \epsilon\,.$$
Applying Cauchy-Schwarz and~\eqref{eq:uwf_2} to the last term in~\eqref{eq:uwf_3} we
obtain~\eqref{eq:uwf_1}.

\section{Additive rate-distortion for quantization of a colored $X$}\label{sec:colo_x}

It is worth mentioning that the rate-distortion region~\eqref{eq:rd_rc_G_intro} we obtained also characterizes the tradeoff between rate and distortion of a particular quantization scheme in a different, but closely related setup.

In particular, let $X\sim\m{N}(0,\Sigma_X)$ be a Gaussian vector in $\RR^n$ to be quantized under the standard quadratic distortion measure $D=\frac{1}{n}\EE\|\hat{X}-X\|_2^2$. Clearly the optimal rate-distortion tradeoff for this problem is given by
\begin{align}
R(D)=\frac{1}{n}\min I(X;\hat{X})    
\label{eq:inf_Rd}
\end{align}
where the minimum is over all $P_{\hat{X}|X}$ for which $\frac{1}{n}\EE\|\hat{X}-X\|_2^2\leq D$. As we already discussed, the optimal $P_{\hat{X}|X}$ is determined by the reverse waterfilling solution, and is given precisely by~\eqref{eq:wf_cruve}. In fact, we derived the oracle lower bound by showing that in the oracle setup, where the decoder also knows $\Sigma_X$, our problem is equivalent to that of quantizing $X$ under standard quadratic loss.

While the waterfilling solution gives the optimal tradeoff $R^*(D)$ function, any other valid choice of $P_{\hat{X}|X}$ yields an achievable $R(D)$. A very simple choice is to construct $\hat{X}$ by first adding independent noise $Z\sim\m{N}(0,\frac{1}{T} I_n)$ to $X$ and then performing minimum mean squared error (MMSE) estimation  of $X$ from $X+Z$. The $R(D)$ tradeoff attained by this particular choice of $P_{\hat{X}|X}$ has been studied in the information theory literature (for sources that are not necessarily Gaussian) under the name \emph{additive rate-distortion function (ARDF)}~\cite{zamir2002rate,zamir2002multiterminal,oz11}.

The MMSE estimator of $X$ from $X+Z$ is linear and therefore $\hat{X}=F(X+Z)$, where
\begin{align}
F=\Sigma_X\left(\Sigma_X+\frac{1}{T}I_n \right)^{-1} =U\cdot\diag\left(\frac{\lambda_1 T}{1+\lambda_1 T},\ldots,\frac{\lambda_n T}{1+\lambda_n T} \right)\cdot U^\top   
\end{align}
and 
\begin{align}
\frac{1}{n}\EE\|\hat{X}-X\|_2^2=\frac{1}{n}\sum_{i=1}^n \frac{\lambda_i}{1+\lambda_i T}=    \Drc(\lambda, T).
\end{align}
Furthermore, if all singular values of $\Sigma_X$ are positive, $F$ is invertible and
\begin{align}
\frac{1}{n}I(X;\hat{X})=\frac{1}{n}I(X;X+Z)=\frac{1}{n}\sum_{i=1}^n\frac{1}{2}\log(1+\lambda_i T)=\Rrc(\lambda,T).    
\end{align}
Thus, the $\Sigma_X$ universal rate distortion tradeoff we derived for the problem of quantizing a white source $\m{N}(0,I_n)$ under $d_{\Sigma_X}$ metric known only to the encoder is precisely the additive rate-distortion function for quantizing $X\sim\m{N}(0,\Sigma_X)$ under standard quadratic loss.



\end{document}